\documentclass[11pt,a4paper]{article}
\usepackage{color}
\usepackage[latin1]{inputenc}
\usepackage[T1]{fontenc}
\usepackage[english]{babel}
\usepackage{amssymb}
\usepackage{graphicx}
\usepackage{dsfont}
\usepackage{epsfig}
\usepackage{ae}
\usepackage{hyperref}
\usepackage{harvard}
\usepackage{amsmath}
\usepackage{setspace}
\usepackage{amsthm}
\usepackage{appendix}

\def\Q{\hbox{\rm Q}}

\newcommand{\E}{\mathds{E}}

\newcommand{\si}{\sigma}

\renewcommand{\(}{\left(}
\renewcommand{\)}{\right)}

\renewcommand{\proof}{{\noindent \it Proof. }}
\newtheorem{theo}{Theorem}[section]
\newtheorem{pr}{Proposition}[section]
\newtheorem{lem}{Lemma}[section]
\newtheorem{co}{Corollary}[section]
\newtheorem{re}{Remark}[section]
\newtheorem{defi}{Definition}[section]

\newcommand{\be}{\begin{eqnarray}}
\newcommand{\ee}{\end{eqnarray}}
\newcommand{\by}{\begin{eqnarray*}}
\newcommand{\ey}{\end{eqnarray*}}
\newcommand{\bt}{\begin{theo}}
\newcommand{\et}{\end{theo}}
\newcommand{\bl}{\begin{lem}}
\newcommand{\el}{\end{lem}}
\newcommand{\bc}{\begin{co}}
\newcommand{\ec}{\end{co}}
\newcommand{\eex}{\end{exa}\vspace{-3mm}}
\newcommand{\br}{\begin{re}}
\newcommand{\er}{\end{re}\vspace{-3mm}}
\renewcommand{\geq}{\geqslant}
\renewcommand{\leq}{\leqslant}
\renewcommand{\ge}{\geqslant}
\renewcommand{\le}{\leqslant}

\citationstyle{agsm}

\begin{document}

\onehalfspacing
\countdef\parno=11
\parno=1
\def\exhibit#1{{\begin{flushleft}\large \textsc{  Exhibit}  \the\parno: \newline \rm\bf\small  #1  \end{flushleft}}\global\advance\parno by 1
\nobreak}

\title{Prices and Asymptotics  \\
for Discrete Variance Swaps}
\author{\textsc{Carole Bernard} \footnote{ C.  Bernard is with the department of Statistics and
Actuarial Science at the University of Waterloo,  Email: \texttt{c3bernar@uwaterloo.ca}.  C. Bernard acknowledges support from the Natural Sciences and Engineering Research Council of Canada. }\  and  \textsc{Zhenyu Cui} \footnote{ Corresponding author. Zhenyu
Cui is a Ph.D. candidate in Statistics at the
University of Waterloo,   Email: \texttt{cuizhyu@gmail.com}.  Z. Cui acknowledges support from the Bank of Montreal Capital Markets Advanced Research Scholarship.} \footnote{Both authors thank Jinyuan Zhang for her help as research assistant, as well as seminar participants Christa Cuchiero,  Olympia Hadjiliadis, Antoine Jacquier, Adam Kolkiewicz, Roger Lee, Don McLeish,  Johannes Ruf, and David Saunders for helpful suggestions. We are particularly grateful to an anonymous referee for his/her  constructive and helpful comments on an earlier draft of this paper.}}
\date{Draft: \today }
\maketitle
\begin{abstract}
We study the fair strike of a discrete variance swap for a general time-homogeneous stochastic volatility model. In the special cases of Heston, Hull-White and Sch\"obel-Zhu stochastic volatility  models we give simple explicit expressions (improving Broadie and Jain  \citeyear{BJ08} in the case of the Heston model). We give conditions on parameters under which the fair strike of a discrete variance swap is higher or lower than that of the continuous variance swap. The interest rate and the correlation between the underlying price and its volatility are key elements in this analysis. We derive asymptotics for the discrete variance swaps and  compare our results with those  of Broadie and Jain \citeyear{BJ08}, Jarrow et al. \citeyear{J12} and Keller-Ressel and Griessler \citeyear{KG12}.
\end{abstract}
\textbf{Key-words:} Discrete Variance swap, Heston model, Hull-White model, Sch\"{o}bel-Zhu model.
\newpage

\begin{center}
{\bf \Large Prices and Asymptotics  \\\vspace{3mm}
for Discrete Variance Swaps}
\end{center}

\vspace{1.5cm}

\section{Introduction}

 A variance swap is a  derivative contract that pays at a fixed maturity $T$ the difference between a given level (fixed leg) and a   realized level of variance over the swap's life (floating leg). Nowadays, variance swaps on stock indices are broadly used and  highly liquid. Less standardized variance swaps could be linked to other types of underlying stocks such as currencies or commodities. They can be useful for hedging volatility risk exposure or for taking positions on future realized volatility. For example, Carr and Lee \citeyear{CL07} price options on realized variance and realized volatility by using  variance swaps as pricing and hedging instruments. See  Carr and Lee \citeyear{CL} for a history of volatility derivatives. As noted by Jarrow et al. \citeyear{J12}, most academic studies\footnote{See, for example, Howison, Rafailidis and Rasmussen \citeyear{HRR04}, Windcliff, Forsyth and Vetzal \citeyear{WFV06}, Benth, Groth and Kufakunesu \citeyear{BGK07} and Broadie and Jain \citeyear{BJ08b}.} focus on  continuously sampled variance and volatility swaps. However, existing volatility derivatives tend to be based on the realized variance computed from the discretely sampled log stock price and continuously sampled derivatives prices may only be  used as approximations.  As pointed out in Sepp \citeyear{S}, some care is needed to replace the  discrete realized variance by the continuous quadratic variation.  By standard probability arguments,  the discretely sampled realized variance  converges to the quadratic variation of the log stock price in probability. However, this does not guarantee that it converges in expectation.  Jarrow et al. \citeyear{J12}  provide sufficient conditions such that the convergence in expectation  happens when the stock is modeled by a general semi-martingale, and concrete examples where this convergence fails.

 In this paper we study discretely sampled variance swaps in a general time-homogeneous model for  stochastic volatility. For discretely sampled variance swaps, it is difficult to use the elegant and model-free approach of Dupire \citeyear{D93} and Neuberger \citeyear{N94}, who independently proved that the fair strike for a continuously sampled variance swap on any underlying price process with continuous path is simply two units of the forward price of the log contract. Building on these results, Carr and Madan \citeyear{CM} published an explicit expression to obtain  this forward price  from option prices (by synthesizing a forward contract with vanilla options). The Dupire-Neuberger theory was recently extended by Carr, Lee and Wu \citeyear{CLW11} to the case when the underlying stock price is driven by a time-changed L\'evy process (thus allowing jumps in the path of the underlying stock price). In this paper, we adopt a parametric approach that allows us to derive explicit closed-form expressions and asymptotic behaviors with respect to key parameters such as the maturity of the contract, the risk-free rate, the sampling frequency, the volatility of the variance process, or the correlation between the underlying stock and its volatility.  This is in line with the work of Broadie and Jain \citeyear{BJ08} in which the Heston model and the Merton jump diffusion model are considered. See also Itkin and Carr \citeyear{IC} who study discretely sampled variance swaps in the 3/2 stochastic volatility model.

 Our main contributions are as follows. We give an expression of  the fair strike of the discretely sampled variance swap and derive its sensitivity to the interest rate in a general time-homogeneous stochastic volatility model. In the case of the (correlated) Heston \citeyear{H} model, the (correlated) Hull-White \citeyear{HW} model, and the (correlated) Sch\"{o}bel-Zhu \citeyear{SZ99} model, we obtain simple explicit closed-form formulas for the respective fair strikes of continuously and discretely sampled variance swaps. In the Heston model, our formula simplifies the  results of Broadie and Jain \citeyear{BJ08} and is easy to analyze. Consequently, we are able to give asymptotic behaviors with respect to key parameters of the model and to the sampling frequency. 
 In particular, we provide explicit conditions under which the fair strike of  the discretely sampled variance swap is less valuable than that of the continuously sampled variance swap for high sampling frequencies, although the contrary is commonly  observed in the literature (see B\"{u}hler \citeyear{B} for example). Thus the ``convex-order conjecture" formulated by Keller-Ressel and Griessler \citeyear{KG12} may not hold for stochastic volatility models with correlation.  We discuss practical implications and illustrate the risk to underestimate or overestimate prices of  discretely sampled variance swaps when using a model for the corresponding continuously sampled ones with numerical examples. 

 The paper is organized as follows. Section \ref{S1} deals with the general time-homogeneous stochastic volatility model. Sections \ref{S2}, \ref{S3} and \ref{S3bis} provide formulas for the fair strike of a discrete variance swap in the Heston, the Hull-White and the Sch\"{o}bel-Zhu models. Section \ref{S5} contains  asymptotics for the Heston, the Hull-White and the Sch\"{o}bel-Zhu models and discusses the ``convex-order conjecture". A numerical analysis is given in Section \ref{S4}.

\section{Pricing Discrete Variance Swaps in a Time-homogeneous Stochastic Volatility Model\label{S1}}
In this section, we consider the problem of pricing a discrete variance swap under the following general time-homogeneous stochastic volatility model $(M)$, where the stock price and its volatility can possibly be correlated. We assume a constant risk-free rate $r\ge 0$, and that under a risk-neutral probability measure $\Q$
\begin{equation}(M)\quad
\left\{\begin{array}{rcl}
\frac{dS_t}{S_t} &=& r dt +m(V_t) dW_t^{(1)}\\
dV_t& = &\mu (V_t) dt +\sigma(V_t)dW_t^{(2)}
\end{array}\right.\label{eq1}\end{equation}
where $\E[dW_t^{(1)}dW_t^{(2)}]=\rho dt$, with $W^{(1)}$,  $W^{(2)}$ standard correlated Brownian motions. The state space of the stochastic variance process $V$ is $J=(l,r), -\infty\leq l<r\leq \infty$. Assume that  $\mu, \sigma: J\rightarrow \mathbb{R}$ are Borel functions satisfying the following Engelbert-Schmidt conditions,
$\forall x\in J, \sigma(x)\neq 0,$  $\frac{1}{\sigma^2(x)}, \frac{\mu(x)}{\sigma^2(x)}, \frac{m^2(x)}{\sigma^2(x)} \in L_{loc}^1 (J)$.
 Here $L_{loc}^1 (J)$ denotes the class of locally integrable functions, i.e. the functions $J\rightarrow \mathbb{R}$ that are integrable on compact subsets of $J$. Under the above conditions, the SDE \eqref{eq1} for $V$ has a unique in law weak solution that possibly exits its state space $J$ (see Theorem $5.5.15$, page 341, Karatzas and Shreve \citeyear{KS91}). Assume that $ \frac{m(x)}{\sigma(x)}$ is differentiable at all $x\in J$.

In particular, this general model includes the Heston, the Hull-White, the Sch\"{o}bel-Zhu, the $3/2$  and the Stein-Stein models as special cases.  In what follows, we study  discretely and continuously sampled variance swaps with maturity $T$. In a variance swap, one counterparty agrees to pay at a fixed maturity $T$ a notional amount times the difference between a fixed level and a realized level of variance over the swap's  life. If it is continuously sampled, the realized variance corresponds to the quadratic variation of the underlying log price. When it is discretely sampled, it is the sum of the squared increments of the log price. Define their respective fair  strikes  as follows.

\begin{defi}
The fair strike of the discrete variance swap associated with the partition $0=t_0 <t_1 <... < t_n =T$ of the time interval $[0,T]$ is defined as
\begin{align}
K^M_d(n) &:=
\frac{1}{T}\sum_{i=0}^{n-1} \E\left[\(\ln \frac{S_{t_{i+1}}}{S_{t_i}}\)^2   \right],\label{KD}
\end{align}
where the underlying stock price $S$ follows the time-homogeneous stochastic volatility model \eqref{eq1} and where the exponent $M$ refers to the model $(M)$.
\end{defi}
\begin{defi} The fair strike of the continuous variance swap  is defined as
\begin{align}
K^M_c &:=\frac{1}{T}\E\left[\int_0^T m^2(V_s) ds\right],\label{KC}
\end{align}
where $S$ follows the time-homogeneous stochastic volatility model \eqref{eq1}.
\end{defi}
 In popular stochastic volatility models, $m(v)=\sqrt{v}$, so that $K^M_c=\frac{1}{T}\E\left[\int_0^T V_s ds\right]$. The derivation of the fair strike of a discrete variance swap in the time-homogeneous stochastic volatility model \eqref{eq1} is based on the following proposition.
\begin{pr}\label{p1}
Under the dynamics \eqref{eq1} for the stochastic volatility model $(M)$, define
\be \label{fh}f(v)=\int_{0}^v \frac{m(z)}{\sigma(z)}dz\quad\hbox{and}\quad
h(v)=\mu(v)f^{\prime}(v)+\frac{1}{2}\sigma^2 (v)f^{\prime\prime}(v).\notag\ee

For all $t\leq s\leq t+\Delta$ and $t\leq u \leq t+\Delta$, assume that\footnote{These conditions ensure that we can apply Fubini's theorem to exchange the order of integration. They are easily verified in specific examples.}
\begin{eqnarray}
\E\left[ \left|h(V_s) h(V_u)\right|\right]<\infty,\quad \E\left[ \left|h(V_s) m^2(V_u)\right|\right]<\infty,\notag\\
\E\left[\left|(f(V_{t+\Delta})-f(V_t))(2\rho h(V_s)+m^2(V_s))\right|\right]<\infty,\label{tech}
\end{eqnarray}

\newpage
Define for $t\leq s\leq t+\Delta,\quad  t\leq u \leq t+\Delta$,

\begin{tabular}{ll}
$m_1(s):=\E\left[ m^2(V_s)\right]$, &\quad
$m_2(s,u):=\E\left[m^2(V_s)m^2(V_u)\right]$,\\
$m_3(s,u):=\E\left[ h(V_s) h(V_u)\right]$,&\quad
$m_4(s,u):=\E\left[ h(V_s) m^2(V_u)\right]$, \\
\multicolumn{2}{l}{$m_5(t,s):=\E\left[(f(V_{t+\Delta})-f(V_t))(2\rho h(V_s)+m^2(V_s))\right], then$}
\end{tabular}
\vspace{3mm}
\begin{align}
&\E\left[\(\ln \frac{S_{t+\Delta}}{S_t}\)^2\right]=r^2 \Delta^2
+(1-\rho^2 -r\Delta) \int_t^{t+\Delta} m_1(s)ds- \rho \int_t^{t+\Delta} m_5(t,s) ds\notag\\
&\quad \quad +\frac{1}{4}  \int_t^{t+\Delta}\int_t^{t+\Delta} m_2(s,u)dsdu+ \rho^2 \E\left[\(f(V_{t+\Delta})-f(V_t)\)^2\right]\notag\\
&\quad \quad+ \rho^2   \int_t^{t+\Delta}\int_t^{t+\Delta} m_3(s,u)dsdu
 +\rho \int_t^{t+\Delta}\int_t^{t+\Delta} m_4(s,u)dsdu.\label{g4}
\end{align}

\end{pr}
\proof See Appendix \ref{pp1}.\hfill$\Box$

Proposition \ref{p1} gives the key equation in the analysis of  discrete variance swaps. Observe\footnote{Thanks to the anonymous referee for pointing out this general expression.} that the final expression \eqref{g4} only depends on covariances of  functionals of $V_t$. Thus we can derive closed-form formulas for the fair strike of discrete variance swaps in those stochastic volatility models in which the terms $m_i$ from Proposition \ref{p1} can be computed in closed-form.
In the rest of the paper, we provide three examples  to apply this formula.

     From now on, for simplicity, we consider the equi-distant sampling scheme in \eqref{KD}. Under this scheme, $t_{i}=iT/n$ and $\Delta= t_{i+1}-t_i =T/n$, for $i=0,1,...,n$.

\begin{re}\label{INTR}
From \eqref{g4} it is clear that the fair strike of a discrete variance swap only depends on the risk-free rate $r$ up to the second order, as there is no higher order terms of $r$.  Interestingly, the second order coefficient of this expansion is model-independent whereas the first order coefficient is directly related to the strike of the corresponding continuously-sampled variance swap. Assume a constant sampling period $\frac{T}{n}$, the fair strike of the discrete variance swap can be expressed as
\begin{align}
K_d^M(n)&= b^M (n) -\frac{T}{n} K_c^M r +\frac{T}{n}r^2,\label{formr}
\end{align}
where $b^M (n)$ does not depend on $r$. Its sensitivity\footnote{The impact of stochastic interest rates on variance swaps is studied by H\"orfelt and Torn\'e \citeyear{HO}. Long-dated variance swaps will usually be sensitive to the interest rate volatility.}  to the risk-free rate $r$ is equal to
\begin{align}
\frac{d K_d^M (n)}{dr}=\frac{T}{n}(2r-K_c^M).\label{rhoS}
\end{align}
so that the minimum of $K_d^M$ as a function of $r$ is attained when the risk-free rate takes the value $r^*$ given by
\be r^*=\frac{K_c^M}{2}.\notag\ee

\end{re}

 The next proposition deals with the special case when the risk-free rate $r$ and the correlation coefficient $\rho$ are both equal to 0.
\begin{pr}(Fair strike when $r=0\%$ and $\rho=0$)\label{p0}

In the special case when the constant risk-free rate is 0, and the underlying stock price is not correlated to its volatility, we observe that
$$K^M_d(n)\ge K^M_c.$$
\end{pr}
\proof Using Proposition \ref{p1} when  $r=0\%$ and $\rho=0$, we obtain
\begin{align}
\E\left[\(\ln \frac{S_{t+\Delta}}{S_t}\)^2\right] &=
\frac{1}{4}\E\left[  \(\int_t^{t+\Delta}m^2(V_s)ds\)^2\right]+ \int_t^{t+\Delta} \E\left[ m^2(V_s) \right]ds.\notag
\end{align}
We then add up the expectations of the squares of the log increments (as in \eqref{KD}) and find that the fair strike of the discrete variance swap is always larger than the fair strike of the continuous variance swap given in \eqref{KC}.
\hfill$\Box$

Proposition \ref{p0} has already appeared in the literature in specific models.  See for example Corollary $6.2$ of  Carr, Lee and Wu \citeyear{CLW11}, where this result is proved in the more general setting of time-changed L\'evy processes with independent time changes. However, we will see in the remainder of this paper that Proposition \ref{p0} may not hold under more general assumptions, namely when the dynamic of the stock price is correlated to the one of the volatility.

%

\section{Fair Strike of the Discrete Variance Swap in the Heston model\label{S2}}
Assume that we work under the Heston stochastic volatility model with the following dynamics
\begin{equation}(H) \quad
\left\{\begin{array}{rl}
\frac{dS_t}{S_t} &=r  dt + \sqrt{V_t} dW^{(1)}_t,\\
dV_t &=\kappa(\theta-V_t)dt+\gamma \sqrt{V_t} dW_t^{(2)}
\end{array}\right.\label{eqv}
\end{equation}
where $\E\left[dW^{(1)}_t dW^{(2)}_t\right]=\rho dt$. It is a special case of the general model \eqref{eq1}, where we choose
\begin{align}
m(x)&=\sqrt{x},\ \mu(x)=\kappa(\theta-x),\ \sigma(x)=\gamma \sqrt{x}.\label{inter}
\end{align}
Using \eqref{eqgsv} in Lemma \ref{l1} in the Appendix with $f(v)=\frac{v}{\gamma}$ and $h(v)=\frac{\kappa}{\gamma}(\theta-v)$, the stock price is
\begin{align}
S_t &=S_0 e^{rt-\frac{1}{2}\xi_t +(V_t -V_0-\kappa\theta t+\kappa \xi_t)\frac{\rho}{\gamma} +\sqrt{1-\rho^2} \int_{0}^{t}\sqrt{V_s} dW^{(3)}_s}\label{eqs}
\end{align}
where $\xi_t=\int_0^tV_sds$ and $W^{(3)}_t$ is such that $dW^{(1)}_t=\rho dW^{(2)}_t+\sqrt{1-\rho^2}dW^{(3)}_t$.

Using Proposition \ref{p1} for the time-homogeneous stochastic volatility model, we then derive a closed-form expression for the fair strike of a discrete variance swap and compare it with the fair strike of a continuous variance swap.
\begin{pr}(Fair Strikes in the Heston Model) \label{p2}

In the Heston stochastic volatility model \eqref{eqv},
the fair strike \eqref{KD} of the discrete variance swap is
\begin{eqnarray}
K^H_d(n)=\frac{1}{8{n}{\kappa}^{3} T }
\left\{
n\left(\gamma^2 \left(\theta-2V_0 \right)+2\kappa \left(V_0 -\theta\right)^2\right) \left(e^{-2\kappa T}-1\right) \frac{1-e^{\frac{\kappa T}{n}} }{1+e^{\frac{\kappa T}{n}}}
\right.\notag\\
\left.
+2\kappa T   \left( \kappa^2 T\left(\theta-2r\right)^2 +n\theta \left(4\kappa^2 -4\rho\kappa \gamma +\gamma^2\right)\right)
\right. \label{discrete}\\
\left.
+4 \left(V_0 -\theta \right) \left(n\left(2\kappa^2 +\gamma^2 -2\rho\kappa\gamma\right)+\kappa^2 T\left(\theta-2r\right)  \right) \left(1-e^{-\kappa T}\right)
\right.\notag\\
\left.
-2n^2\theta \gamma\left(\gamma-4\rho\kappa \right) \left(1-e^{-\frac{\kappa T}{n}}\right)
+4\left(V_0 -\theta\right)\kappa T \gamma\left(\gamma-2\rho \kappa \right) \frac{1-e^{-\kappa T}}{1-e^{\frac{\kappa T}{n}}}
 \right\},\notag\end{eqnarray}
where $a=r-\frac{\rho\kappa\theta}{\gamma}$ and $b=\frac{\rho\kappa}{\gamma} -\frac12$. The fair strike of the continuous variance swap is
\begin{align}
K^H_c&=\frac{1}{T}\E\left[\int_0^T V_s ds\right]=\theta  +(1-e^{-\kappa T}) \frac{V_0 -\theta}{\kappa T}.\label{cont}
\end{align}

\end{pr}
\proof See Appendix \ref{p2p} for the proof of \eqref{discrete}. The formula \eqref{cont} for the fair strike of a continuous variance swap is already well-known and can be found for example in  Broadie and Jain \citeyear{BJ08},  formula $(4.3)$, page 772.\hfill$\Box$

Proposition \ref{p2} provides an explicit formula for the fair strike of a discrete variance swap as a function of model parameters. This formula simplifies the expressions  obtained by  Broadie and Jain \citeyear{BJ08} in equations (A-29) and (A-30), page 793, where several sums from 0 to $n$ are involved and can actually be computed explicitly as shown by the expression \eqref{discrete} above. We verified that our formula agrees with numerical examples presented in Table 5 (column `SV') on page 782 of Broadie and Jain \citeyear{BJ08}.\footnote{This formula has been implemented in Matlab and its code is available upon request from authors as well as for all other formulas that appear in this paper.}

Contrary to what is stated in the introduction of the paper by  Zhu and Lian \citeyear{ZL}, the techniques of Broadie and Jain \citeyear{BJ08} can easily be extended to other types of payoffs.  The following proposition gives explicit expressions for the volatility derivative considered by Zhu and Lian \citeyear{ZL}.
\begin{pr}\label{zl} For the following set of dates $t_i=\frac{iT}{n}$ with $i=0,1,...,n$, denote $\Delta=T/n$, and assume $\alpha=2\kappa\theta/\gamma^2 -1 \geq 0$, and $\gamma^2 T<1$.  Then the fair price of a discrete variance swap with payoff $\frac{1}{T}\sum\limits_{i=0}^{n-1}\(\frac{S_{t_{i+1}}-S_{t_{i}}}{S_{t_{i}}}\)^2$ is equal to
\begin{align*}
K_d^{zl}(n)&=\frac{1}{T} \sum_{i=0}^{n-1} \E\left[ \(\frac{S_{t_{i+1}}-S_{t_{i}}}{S_{t_{i}}}\)^2  \right]=\frac{1}{T}\(a_0 +\sum_{i=1}^{n-1} a_i\)+\frac{n-2ne^{r\Delta}}{T},
\end{align*}
where we define $a_i =\E\left[\(\frac{S_{t_{i+1}}}{S_{t_{i}}}\)^2\right]$, for $i=0,1,...,n-1$. Then for $i=0,1,...,n-1$, we have
\begin{equation}
a_i=\frac{e^{2r\Delta}}{S_0^2} M(2, \Delta) e^{q(2) V_0 \(\frac{\eta(t_i)e^{-\kappa t_i} }{\eta(t_i) -q(2)}-1\)} \(\frac{\eta(t_i)}{\eta(t_i) -q(2)}\)^{\alpha+1},\notag
\end{equation}
where
\begin{align*}
M(u,t)&=\E[e^{u X_t }]=S_0^ue^{\frac{\kappa\theta}{\gamma^2} \( (\kappa-\gamma \rho u-d(u))t-2\ln \left(\frac{1-g(u) e^{-d(u) t}}{1-g(u)}\right)  \)}e^{ V_0 \frac{\kappa-\gamma \rho u-d(u)}{\gamma^2} \frac{1-e^{-d(u) t}}{1-g(u) e^{-d(u) t}}},\notag
\end{align*}
with the following auxiliary functions
\begin{align}
d(u)&=\sqrt{(\kappa-\gamma \rho u)^2 +\gamma^2 (u-u^2)  }, &
g(u)=\frac{\kappa-\gamma \rho u-d(u)}{\kappa-\gamma \rho u+d(u)},\notag\\
q(u)&=\frac{\kappa-\gamma \rho u-d(u)}{\gamma^2} \frac{1-e^{-d(u)\Delta}}{1-g(u)e^{-d(u)\Delta}}, &
\eta(u)=\frac{2\kappa}{\gamma^2} \(1-e^{-\kappa u}\)^{-1}.\notag
\end{align}
\end{pr}
\begin{proof}
See Appendix \ref{proofzl}.
\hfill$\Box$
\end{proof}
\begin{re}
The formula in the above Proposition \ref{zl} is consistent with  the one obtained in equation (2.34) by Zhu and Lian \citeyear{ZL}. In particular, we are able to reproduce all numerical results but one presented in Table $3.1$, page $246$ of Zhu and Lian \citeyear{ZL} using their set of parameters: $\kappa=11.35$, $\theta=0.022$, $\gamma=0.618$, $\rho=-0.64$, $V_0=0.04$, $r=0.1$, $T=1$ and $S_0 =1$ (all numbers match except the case when $n=4$ we get 263.2 instead of 267.6).


Proposition \ref{zl} gives a formula for pricing the variance swap with payoff  $\frac{1}{T}\sum\limits_{i=0}^{n-1}\(\frac{S_{t_{i+1}}-S_{t_{i}}}{S_{t_{i}}}\)^2$, but it is straightforward to extend its proof to the following payoff $\frac{1}{T}\sum\limits_{i=0}^{n-1}\(\frac{S_{t_{i+1}}-S_{t_{i}}}{S_{t_{i}}}\)^k$, with an arbitrary integer power $k$.
\end{re}

\section{Fair Strike of the Discrete Variance Swap in the Hull-White model\label{S3}}
The correlated Hull-White stochastic volatility model is as follows
\begin{equation}(HW) \quad
 \left\{\begin{array}{cl}
\frac{dS_t}{S_t} &= r dt +\sqrt{V_t} dW_t^{(1)}\\
dV_t& =\mu V_t dt +\sigma V_t dW_t^{(2)}
\end{array}\right.\label{hwmodel}
\end{equation}
where $\E[dW^{(1)}_t dW^{(2)}_t]=\rho dt$. Referring to equation \eqref{eq1}, we have $m(x)=\sqrt{x},\ \mu(x)=\mu x,\ \sigma(x)=\sigma x$,
so it is straightforward to determine
$f(v)=\frac{2}{\sigma}\sqrt{v}$,  $h(v)=\(\frac{\mu}{\sigma}-\frac{\sigma}{4}\)\sqrt{v}$, and apply \eqref{eqgsv} in Lemma \ref{l1} in the Appendix to obtain
\begin{multline}
S_T=S_0 \exp\left\{ rT -\frac{1}{2}\int_0^T V_t dt+\frac{2\rho}{\sigma}(\sqrt{V_T}-\sqrt{V_0})\right.\\
\left.-\rho \(\frac{\mu}{\sigma}-\frac{\sigma}{4}\)\int_0^T \sqrt{V_t} dt + \sqrt{1-\rho^2} \int_0^T
\sqrt{V_t} dW_t^{(3)} \right\}.\notag
\end{multline}

\begin{pr}(Fair Strikes in the Hull-White Model) \label{pHW}

In the Hull-White stochastic volatility model \eqref{hwmodel},
the fair strike \eqref{KD} of the discrete variance swap is
\begin{multline}
K^{HW}_d(n)= {\frac {{r}^{2}{T}}
{n}}+ \left( 1-{\frac {rT}{n}} \right)  K^{HW}_c -\frac{V_0^{2}\(e^{\left( 2\,\mu+
{\sigma}^{2} \right)T }-1\)\left(e^{\frac{\mu T}{n}}-1 \right) }{2T\mu(\mu+\sigma^2)\(e^{{\frac {\left( 2\,\mu+
{\sigma}^{2} \right)T }{n}}}-1\)}\\
+\frac{V_0^{2}\(e^{\left( 2\,\mu+
{\sigma}^{2} \right)T }-1\) }{2T(2\mu+\sigma^2 )(\mu+\sigma^2)}
+\frac{8\rho
 \left( {e^{\frac{3(4\mu+\sigma^2)T}{8} }}-1 \right) {V_0}^{3/2} \sigma (e^{\frac{\mu T}{n}}-1) }{\mu  T \left( 4\,\mu+3\,{
\sigma}^{2} \right)\left( e^{\frac{3(4\mu+\sigma^2)T}{8n}}-1 \right)}\\-\frac{64\rho
 \left( {e^{\frac{3(4\mu+\sigma^2)T}{8} }}-1 \right) {V_0}^{3/2} \sigma}{ 3 T(4\mu+\sigma^2) \left( 4\,\mu+3\,{
\sigma}^{2} \right)}
.\label{discretehw}
\end{multline}

The fair strike of the continuous variance swap is
\begin{align}
K^{HW}_c&=\frac{1}{T}\E\left[\int_0^T V_s ds\right]=\frac{V_0}{T\mu} (e^{\mu T}-1).\label{contHW}
\end{align}

\end{pr}
\proof The proof can be found in Appendix \ref{proofHW}.\hfill$\Box$

\section{Fair Strike of the Discrete Variance Swap in the Sch{\"o}bel-Zhu model\label{S3bis}}

The correlated Sch{\"o}bel-Zhu stochastic volatility model (see Sch{\"o}bel and Zhu \citeyear{SZ99}) can be described by the following dynamics\footnote{\label{f6}We shall note that here $m(V_t)=V_t$ (where $m(\cdot)$ is defined in \eqref{eq1}) instead of $\sqrt{V_t}$, thus the process $V_t$ models the volatility and not the variance. In particular in the Sch\"obel-Zhu model, the variance process $Y_t=V_t^2$ follows $dY_t=(\gamma^2+2\kappa\theta\sqrt{Y_t}-2\kappa Y_t)dt+2\gamma\sqrt{Y_t}dW_t^{(2)}$.}
\begin{equation}(SZ) \quad
\left\{\begin{array}{cl}
\frac{dS_t}{S_t} &= r dt +V_t dW_t^{(1)}\\
dV_t& =\kappa(\theta-V_t) dt +\gamma dW_t^{(2)}
\end{array}\right.\label{szmodel}
\end{equation}
where $\E[dW^{(1)}_t dW^{(2)}_t]=\rho dt$. Referring to equation \eqref{eq1}, we have
$m(x)=x, \mu(x)=-\kappa(x-\theta), \sigma(x)=\gamma$,
so it is straightforward to apply \eqref{eqgsv} in Lemma \ref{l1} given in the Appendix with
$f(v)=\frac{v^2}{2\gamma}$ and   $h(v)=\frac{\kappa\theta}{\gamma}v-\frac{\kappa}{\gamma}v^2 +\frac{\gamma}{2}$ to obtain
\begin{multline}
S_T=S_0 \exp\left\{ (r-\frac{\gamma\rho}{2})T -\frac{\kappa\theta\rho}{\gamma}\int_0^T V_t dt-\(\frac{1}{2}-\frac{\rho\kappa}{\gamma}\)\int_0^T V_t^2 dt  \right.\\
\left.+\frac{\rho}{2\gamma}(V_T^2-V_0^2) + \sqrt{1-\rho^2} \int_0^T
V_t dW_t^{(3)} \right\}.\notag
\end{multline}

\begin{pr}(Fair Strikes in the Sch{\"o}bel-Zhu Model) \label{pSZ}

In the Sch{\"o}bel-Zhu stochastic volatility model \eqref{szmodel},
the fair strike \eqref{KD} of the discrete variance swap is computed from \eqref{g4} but does not have a simple expression.\footnote{See Proposition \ref{asympt_n_SZ} for an explicit expansion.}
The fair strike of the continuous variance swap is
\begin{align}
K^{SZ}_c&=\frac{\gamma^2}{2\kappa}+\theta^2+\(\frac{(V_0 -\theta)^2}{2\kappa T} -\frac{\gamma^2}{4\kappa^2 T}\)(1-e^{-2\kappa T})+\frac{2\theta(V_0-\theta)}{\kappa T}(1-e^{-\kappa T}) .\label{contsz}
\end{align}

\end{pr}
\proof The proof can be found in Appendix \ref{proofSZ}.\hfill$\Box$

\section{Asymptotics\label{S5}}
In the time-homogeneous stochastic volatility model, this section presents asymptotics  for the fair strikes of discrete variance swaps in the Heston, the Hull-White and the Sch\"{o}bel-Zhu models based on the explicit expressions derived in the previous sections \ref{S2}, \ref{S3} and \ref{S3bis}.

The expansions as functions of the number of sampling periods $n$ are given in Propositions \ref{asympt_n_Heston},  \ref{asympt_n_HW} and \ref{asympt_n_SZ} (respectively for the Heston, Hull-White and Sch\"obel-Zhu models). In the Heston model, our results are consistent with Proposition 4.2 of Broadie and Jain \citeyear{BJ08}, in which it is proved that $K^{H}_d(n)=K^{H}_c+\mathcal{O}\(\frac{1}{n}\)$. The expansion below is more precise in that at least  the first leading term in the expansion is given explicitly. See also Theorem 3.8 of Jarrow et al. \citeyear{J12} in a more general context. In particular, Jarrow et al. \citeyear{J12} give a sufficient condition for the convergence of the fair strike of a discrete variance swap to that of a continuously monitored variance swap. In our setting, which is in the absence of jumps, their sufficient condition reduces to $E[\int_0^T m^4(V_s) ds]<\infty$. This latter condition is obviously satisfied in the three examples considered in this paper (the Heston, the Hull-White and the Sch\"obel-Zhu models).

Expansions as a function of the maturity $T$ (for small maturities) are also given in order to complement results of
Keller-Ressel and Muhle-Karbe \citeyear{KM12} (see for example  Corollary $2.7$ which gives qualitative properties of the discretization gap\footnote{See Definition $2.6$ on page $112$ of  Keller-Ressel and Muhle-Karbe \citeyear{KM12}.} as the maturity $T\rightarrow 0$).

\subsection{Heston Model}
We first expand the fair strike of the discrete variance swap with respect to the number of sampling periods $n$.
\begin{pr}(Expansion of the fair strike $K^{H}_d(n)$ w.r.t. $n$) \label{asympt_n_Heston}

 In the Heston model, the expansion of the fair strike of a discrete variance swap, $K_d^H(n)$, is given by
\begin{align}
K^{H}_d (n)&=K^{H}_c+\frac{a_1^{H}}{n} +\mathcal{O}\(\frac{1}{n^2}\),\label{ExpandHeston}
\end{align}
where
$$a_1^{H}=r^2T-rTK_c^H+\(\frac{\gamma(\theta-V_0)}{2\kappa}(1-e^{-\kappa T})-\frac{\theta\gamma T}{2}\)\rho +\(\frac{\theta^2}{4} +\frac{\theta \gamma^2}{8\kappa}\)T+c_1$$
with \begin{equation*}
c_1^{H}=\frac{\left[{\gamma}^{2}\theta-2\kappa (V_0-\theta)^2\right]\left( {e^{-2T\kappa}}-1 \right) +2(V_0-\theta)({e^{-T\kappa}}-1)\left[{\gamma}^{2}({e^{-T\kappa}}-1)-4\kappa\theta\right]}{16{\kappa}^{2}}.\notag
\end{equation*}

\end{pr}
\proof This proposition is a straightforward expansion from \eqref{discrete} in Proposition \ref{p2}.\hfill$\Box$



We know that $K^{H}_d(n)=b^{H}(n)+\frac{T}{n}r(r-K^{H}_c)$ from \eqref{formr} in Remark \ref{INTR}. It is thus clear that $a_1^{H}$ contains all the terms in the risk-free rate $r$ and thus that all the higher terms in the expansion \eqref{ExpandHeston} with respect to $n$ are independent of the risk-free rate.

\begin{re} \label{RH}The first term in the expansion \eqref{ExpandHeston}, $a_1^{H}$, is a linear function of $\rho$. Observe that the coefficient in front of $\rho$, $\frac{\gamma(\theta-V_0)}{2\kappa}(1-e^{-\kappa T})-\frac{\theta\gamma T}{2}$ is negative,\footnote{This can be easily seen from the fact that for all $x>0$, $(\theta-V_0)(1-e^{-x})-\theta x\le \theta(1-e^{-x}-x)<0$, and note that here $x=\kappa T>0$.} so that $a_1^{H}$ is always a decreasing function of $\rho$. We have that $$
a_1^{H} \ge 0\quad\Longleftrightarrow\quad \rho\le \rho_0^{H}
$$
where
$\rho^{H}_0=\frac{r^2T-rTK_c^H +\(\frac{\theta^2}{4} +\frac{\theta \gamma^2}{8\kappa}\)T+c_1^{H}}{-\(\frac{\gamma(\theta-V_0)}{2\kappa}(1-e^{-\kappa T})-\frac{\theta\gamma T}{2}\)}.$
\end{re}%

%
%
%

\begin{pr}(Expansion of the fair strike for small maturity) \label{asympt_T_Heston}

In  the Heston model, $K^{H}_d(n)$ can be expanded when $T\rightarrow 0$ as
\be K^{H}_d(n)=V_0+b_1^H T+b_2^H T^2+\mathcal{O}\left(T^3\right)\label{expandHeston}\ee
where
\begin{align*}
b_1^H
&=\frac{\kappa(\theta-V_0)}{2}+\frac{1}{4n}\((V_0-2r)^2-2\rho V_0\gamma \)\\
b_2^H &=\frac{\kappa^2 (V_0 -\theta)}{6}+\frac{(V_0 -\theta)\kappa (\gamma \rho +2r-V_0)+\frac{\gamma^2 V_0}{2} }{4n}+\frac{\gamma\rho\kappa(V_0 +\theta)-\frac{\gamma^2 V_0}{2} }{12n^2}.\notag
\end{align*}
Note also that
$K^{H}_c= V_0 +\frac{\kappa}{2}\, \left( \theta-V_0 \right) T+\frac{\kappa}{6}
^{2} \left(  V_0 -\theta \right) {T}^{2}+\mathcal{O} \left( {T}^{3} \right) $ and thus
$$K^{H}_d(n)-K^{H}_c=\frac{1}{4n}\((V_0-2r)^2-2\rho V_0\gamma\right)T+\mathcal{O}(T^2).$$
\end{pr}
\proof This proposition is a straightforward expansion from \eqref{discrete} in Proposition \ref{p2}.\hfill$\Box$

Proposition \ref{asympt_T_Heston} is consistent with  Corollary 2.7 [b] on page $113$ of Keller-Ressel and Muhle-Karbe \citeyear{KM12}, where it is clear that the limit of $K_d(n)-K_c$ is $0$ when $T\rightarrow 0$.

Notice that in the case $\rho\le0$, in the Heston model, $K^{H}_d(n)$ is non-negative and decreasing in $n$ as the maturity $T$ goes to 0. However, this property cannot be generalized to all correlation levels as it depends on the sign of $(V_0-2r)^2-2\gamma V_0\rho$.

\begin{pr}(Expression of the fair strike w.r.t. $\gamma$) \label{asympt_gamma_Heston}

In  the Heston model, $K^{H}_d(n)$ is a quadratic function of $\gamma$:
\be K^{H}_d(n)=\frac{1}{8n\kappa^3T}\left(h_0^H+h_1^H\gamma+h_2^H \gamma^2\right),\label{expandHestonGamma}\ee
where
\begin{align*}
h_0^H&= 2\,n\kappa\, \left( V_0-\theta \right) ^{2} \left( {
e^{-2\,\kappa\,T}}-1 \right)  \frac{ 1-{e^{{\frac {\kappa\,T}{n}}}}}{ 1+{e^{{\frac {\kappa\,T}{n}}}}}+2\,
\kappa\,T \left( {\kappa}^{2}T \left( \theta-2\,r \right) ^{2}+4\,{
\kappa}^{2}n\theta \right) \\
&\quad\quad+ 4\left( V_0-\theta \right)
 \left( 2\,{\kappa}^{2}n+{\kappa}^{2}T \left( \theta-2\,r \right)
 \right)  \left( 1-{e^{-\kappa\,T}} \right),\\
h_1^H&=8\rho\kappa \left(n\theta (n-n e^{-\frac{\kappa T}{n}}-\kappa T) -\left( V_0-\theta \right)\left( n
 \left( 1-{e^{-\kappa\,T}} \right) + \kappa T \frac{1-{e^{-\kappa\,T}}}
{1-{e^{{\frac {\kappa\,T}{n}}}}}\right) \right),\\
h_2^H&=
  n \left( \theta-2\,V_0 \right)  \left( {e^{-2\,
\kappa\,T}}-1 \right)  \frac{1-{e^{{\frac {\kappa\,T}{n}}}}}{1+{e^{{\frac {\kappa\,T}{n}}}}}-2\,{n}^{2}\theta
\, \left( 1-{e^{-{\frac {\kappa\,T}{n}}}} \right) \\&\quad\quad+ 4\left( V_0
-\theta \right) \left( n -n e^{-\kappa\,T} +  \kappa\,T \frac{1-{e^{-\kappa\,T}}}
{1-{e^{{\frac {\kappa\,T}{n}}}}}\right) +2\,\kappa\,Tn
\theta.
\end{align*}
\end{pr}

Proposition \ref{asympt_gamma_Heston} shows that the discrete fair strike in the Heston model is a quadratic function of the volatility  of variance $\gamma$. From Figure \ref{F5}, we observe that the discrete fair strikes evolve in a parabolic shape as $\gamma$ varies.

\subsection{Hull-White Model}

\begin{pr}(Expansion of $K^{HW}_d(n)$ w.r.t. $n$) \label{asympt_n_HW}

 In  the Hull-White model, the expansion of the fair strike of  the discrete variance swap, $K^{HW}_d(n)$, is given by
\be K^{HW}_d(n)=K^{HW}_c+\frac{a_1^{HW}}{n}+\frac{a_2^{HW}}{n^2}+\frac{a_3^{HW}}{n^3} +\mathcal{O}\left(\frac{1}{n^4}\right)\label{expandHW}\ee
where
\begin{align*}
a_1^{HW}&
=r^2 T -r T K^{HW}_c +\frac{V_0^2}{4} \frac{e^{(2\mu+\sigma^2)T}-1}{2\mu +\sigma^2}
-\frac{4\rho \sigma V_0^{\frac{3}{2}}}{3} \frac{e^{\frac{3}{8}(4\mu+\sigma^2)T}-1}{4\mu+\sigma^2},\\
a_2^{HW}&
 =-\frac{V_0^2 \sigma^2 T}{24}\frac{e^{(2\mu+\sigma^2)T}-1}{2\mu +\sigma^2}-\frac{\rho V_0^{\frac{3}{2}} \sigma T (4\mu-3\sigma^2)}{36}\frac{e^{\frac{3}{8}(4\mu+\sigma^2)T}-1}{4\mu+\sigma^2},
\\a_3^{HW}&
=-\frac{\mu T^2 V_0^2 (\mu+\sigma^2)}{48} \frac{e^{(2\mu+\sigma^2)T}-1}{2\mu+\sigma^2}+\frac{\mu T^2 \rho \sigma V_0^{\frac{3}{2}} (4\mu+3\sigma^2)}{72}\frac{e^{\frac{3}{8}(4\mu+\sigma^2)T}-1}{4\mu+\sigma^2}.\notag
\end{align*}
\end{pr}
\proof This proposition is a straightforward expansion from \eqref{discretehw} in Proposition \ref{pHW}.\hfill$\Box$

Observe that $K^{HW}_d(n)=b^{HW}(n)-{\frac{K^{HW}_cT}{n}}r+\frac{T}{n}r^2$ where  $b^{HW}(n)=K^{HW}_d(r=0)>K^{HW}_c$ is independent of $r$.


If we neglect higher order terms in the expansion \eqref{expandHW}, we observe that the position of the fair strike of the discrete variance swap with respect to the fair strike of the continuous variance swap is driven by the sign of $a_1$ and we have the following observation.

\begin{re} \label{RHW}The first term in the expansion \eqref{expandHW}, $a_1^{HW}$, is a linear function of $\rho$. $$a_1^{HW}\ge 0\quad\Longleftrightarrow\quad \rho\le\rho_0^{HW}$$
where
$\rho_0^{HW}
=\frac{3(4\mu+\sigma^2)\(r^2 T -r T K^{HW}_c +\frac{V_0^2}{4} \frac{e^{(2\mu+\sigma^2)T}-1}{2\mu +\sigma^2}\)}{4\sigma V_0^{\frac{3}{2}}(e^{\frac{3}{8}(4\mu+\sigma^2)T}-1)}>0.
$
\end{re}
$\rho_0^{HW}$ can take values strictly larger than 1 as it appears clearly in the right panel of Figure \ref{F3}. In this latter case, the fair strike of the discrete variance swap is larger than the fair strike of the continuous variance swap for all levels of correlation and for  sufficiently high values of $n$. The minimum value of $K_d^{HW}(n)$ as a function of $r$ is obtained when $r=r^*=\frac{K_c^{HW}}{2}$. After replacing $r$ by $r^*$ in the expression of $\rho_0^{HW}$,  $\rho_0^{HW}$ can easily be shown to be positive.\footnote{It reduces to studying the sign of $\frac{e^{(2\mu+\sigma^2) T}-1}{(2\mu+\sigma^2) T}-\frac{(e^{\mu T}-1)^2}{\mu^2T^2}$. It is an increasing function of $\sigma$, so it is larger than $\frac{e^{2\mu T}-1}{2\mu T}-\frac{(e^{\mu T}-1)^2}{\mu^2T^2}$, which is always positive because its minimum is 0 obtained  when $\mu T=0$. }

\newpage
\begin{pr}(Expansion of  $K^{HW}_d(n)$ for small maturity) \label{asympt_T_HW}

In the Hull-White model, $K^{HW}_d(n)$ can be expanded when $T\rightarrow 0$ as
\be K^{HW}_d(n)=V_0+b_1^{HW} T+b_2^{HW} T^2+\mathcal{O}\left(T^3\right),\label{expandHW2}\ee
where
\begin{align*}
b_1^{HW}
&=\frac{ V_0 \,\mu}{2}+\frac{1}{4n}\((V_0-2r)^2-{2\rho{ V_0 }^{3/2}\sigma}\right),\\
b_2^{HW}
&=\frac{ V_0\mu^{2}}{6}+\frac{V_0}{4n}\({\frac {{\sigma}^{2}{ V_0
}}{2}}-{\frac {3\rho\,{ V_0 }^{1/2}{\sigma}(\sigma^2+4\mu)}{8}}+\mu(V_0-2r)\)\\
&\quad\quad\quad+{\frac {{ V_0 }^{3/2}{\sigma}\left(\rho(3\sigma^2-4\mu)-
4\sigma\sqrt{V_0}\right)}{96{n}^{2}}}.\notag
\end{align*}
Note also that
$K^{HW}_c= V_0 +\frac{V_0\mu}{2} T+ \frac{V_0{\mu}^{2}}{6} {T}^{2}+\mathcal{O} \left( {T}^{3} \right), $
and thus
$$K^{HW}_d(n)-K^{HW}_c=\frac{1}{4n}\((V_0-2r)^2-{2\rho{ V_0 }^{3/2}\sigma}\right)T+\mathcal{O}(T^2).$$
\end{pr}
\proof This proposition is a straightforward expansion from \eqref{discretehw} in Proposition \ref{pHW}.\hfill$\Box$

Note that the expansion for small maturities in the Hull White model is similar to the one in the Heston model given in Proposition \ref{asympt_T_Heston}.

\begin{pr}(Expansion of $K^{HW}_d(n)$ w.r.t. $\sigma$) \label{asympt_sigma_HW}

In  the Hull-White model, the fair strike of a discrete variance swap, $K^{HW}_d(n)$, verifies
\be K^{HW}_d(n)=h_0^{HW}+h_1^{HW}\sigma+\mathcal{O}(\sigma^2),\label{expandHWGamma}\ee
where
\begin{align*}
h_0^{HW}&= {\frac {{r}^{2}T}{n}}+ \left( 1-{\frac {rT}{n}} \right) V_0\,
 \frac{{e^{T\mu}}-1}{T\mu}-\frac{{V_0}^{2}}{2}
 \frac{{e^{2\,T\mu}}-1}{ {e^{2\,{\frac {T\mu}{n}}}}-1}  \frac{{e^{{\frac {T\mu}{n}}}}-1}{T\mu^2} +{\frac {{V_0}^{2} \left( {e^{2\,T\mu}}-1
 \right) }{4T{\mu}^{2}}},\\
 h_1^{HW}&=2\rho \frac{ {e^{3/2\,T\mu}}-1}{ {e^{3/2{
\frac {T\mu}{n}}}}-1 } {V_0}^{3/2}\, \frac{{e^
{{\frac {T\mu}{n}}}}-1}{T\mu^2} -{\frac {4\rho \left( {e^{3/2
T\mu}}-1 \right) {V_0}^{3/2}}{3T{\mu}^{2}}}.
\end{align*}
\end{pr}
The expansion of the fair strike in the Hull-White model with respect to the volatility of volatility is very different from the one in the Heston model as it is not a quadratic function of $\sigma$, and it also involves higher order terms of $\sigma$.

\subsection{Sch\"obel-Zhu Model}

We first expand the fair strike of the discrete variance swap with respect to the number of sampling periods $n$. The following result is similar to Proposition \ref{asympt_n_Heston} and \ref{asympt_n_HW}. In particular we find that the first term in the expansion is also linear in $\rho$ and has a similar behaviour as in the Heston and Hull-White model.

\begin{pr}(Expansion of $K_d^{SZ}(n)$ w.r.t. $n$) \label{asympt_n_SZ}

 In the Sch\"obel-Zhu model, the expansion of the fair strike of the discrete variance swap, $K^{SZ}_d (n)$, is given by
\begin{align}
K^{SZ}_d (n)&=K^{SZ}_c+\frac{a_1^{SZ}}{n} +\mathcal{O}\(\frac{1}{n^2}\),\label{ExpandSZ}
\end{align}
where
\begin{align}
a_1^{SZ}&=r^2T-rT K^{SZ}_c+d_1-d_2\frac{\gamma}{2\kappa}\rho,
\end{align}
with
\begin{multline*}
d_1:=\frac{T{V_{{0}}}^{4}}{4}-{\frac {E(T+D)}{16{\kappa}^{2}}}+ \left( \frac{3{V_{{0}}}^{2}{\gamma}^{2}}{4}+{
\frac {E}{32\kappa}}+\frac{\kappa{V_{{0}}}^{3}(\theta-V_0)}{2} \right) {D}^{2}\\
+ \left(\frac{2\theta\,{
\kappa}^{2}{V_{{0}}}^{3}}{3} -\frac{{V_{{0}}}^{4}{\kappa}^
{2}}{6}-\frac{E}{48}-\frac{{V_{{0}}}^{2}{\theta}^{2}{\kappa}^{2}}{2}-{\gamma}^{2}\kappa\,V_{{0}}\theta+\frac{3{V_{{0}}}^{2}\kappa\,{\gamma}^{2}}{4}-\frac{{\gamma}^{4}}{4} \right) {D}^{3}\\
+ \left({\frac {E}{
8\kappa}} +{3{\gamma}^{2}(\theta-V_{{0}})\theta}+\frac{3{V_{{0}}}^{2}{\gamma}^{2}}{2}+V_{0}{\kappa}(\theta-V_0)\left(2\theta^2-\theta V_0+V_0^2\right)\right) \frac{\kappa^2{D}^{4}}{8},
\end{multline*}
and
\begin{multline*}d_2=T \left( {\gamma}^{2}+2\,\kappa\,{\theta}^{2}
 \right)+{{\left( 2\kappa({\theta}^{2}-V_0^2)+
{\gamma}^{2} \right) }}D+\frac{\kappa}{2}
\, \left( {\gamma}^{2}-2\,\kappa\,{(\theta-V_0)}^{2} \right) {D}^{2},
\end{multline*}
where $$E:=4\,{V_{{0}}}^{4}{\kappa}^{2}-4\,{\theta}^{4}{\kappa}^{2}-3\,{\gamma}^{4}-
12\,{\gamma}^{2}{\theta}^{2}\kappa,\quad\quad D:={\frac {{e^{-\kappa\,T}}-1}{\kappa}}.
$$


\end{pr}
\proof This proposition is a straightforward expansion from the formula of $K_d^{SZ}(n)$ in Proposition \ref{pSZ}. Note that although the formula of $K_d^{SZ}(n)$ does not have a simple form,  its asymptotic expansion can be easily computed with Maple for instance. \hfill$\Box$

\begin{re} \label{RSZ} Similarly as in the Heston and the Hull-White models, the first term in the expansion \eqref{ExpandSZ}, $a_1^{SZ}$, is a linear function of $\rho$, but the sign of its slope is not clear in general.
\end{re}%

\begin{pr}(Expansion of the fair strike for small maturity) \label{asympt_T_SZ}

In  the Sch\"obel-Zhu model,  $K^{SZ}_d(n)$ can be expanded when $T\rightarrow 0$ as
\be K_d^{SZ}(n)=V_0^2+b_1^{SZ} T+\mathcal{O}(T^2)\label{expandTSZ}\ee
where
$$b_1^{SZ}=\kappa V_0(\theta-V_0)+\frac{\gamma^2}{2}+\frac{1}{n}\({r^2}-{rV_0^2}+\frac{V_0^2(V_0^2-4\rho\gamma)}{4}\).
$$
Note also that
$K^{SZ}_c= {V_{{0}}}^{2}+ \left(V_{{0}}\kappa(\theta-V_0)+\frac{{\gamma}^{2}}{2} \right) T+\mathcal{O} \left( {T}^{2} \right)
$ and thus,
$$K^{SZ}_d(n)-K^{SZ}_c=\frac{1}{4n}\((V_0^2-2r)^2-4\rho V_0^2\gamma\) T+\mathcal{O}(T^2).$$
\end{pr}

\proof This proposition is a straightforward expansion from the formula of $K_d^{SZ}(n)$ in Proposition \ref{pSZ}.\hfill$\Box$

Note that the form of the expansion is similar for the three models under study (compare Propositions \ref{asympt_T_Heston}, \ref{asympt_T_HW} and \ref{asympt_T_SZ}). We find that the difference between the discrete and the continuous strikes has a first term involving the product of $2\rho$ by a function of the initial variance value and the volatility of the variance process, and respectively $\gamma$ in the Heston, $\sigma$ in the Hull-White and $2\gamma$ in the Sch\"obel-Zhu model. See for example footnote \ref{f6} where the dynamics of the variance is derived in the Sch\"obel-Zhu model.

\subsection{Discussion on the convex-order conjecture}

As motivated in Keller-Ressel and Griessler \citeyear{KG12}, it is of interest to study the \textit{systematic bias} for fixed $n$ and $T$ when using the quadratic variation to approximate the realized variance. B\"{u}hler \citeyear{B} and Keller-Ressel and Muhle-Karbe \citeyear{KM12} show numerical evidence of this bias (see also Section \ref{S4} for further evidence in the Heston and the Hull-White models). Keller-Ressel and Griessler \citeyear{KG12} propose the following ``\textbf{convex-order conjecture}'': $$\E[f(RV(X,\mathcal{P}))]\ge \E[f([X,X]_T)]$$ where $f$ is convex, $\mathcal{P}$ refers to the partition of $[0,T]$ in $n+1$ division points and $X=\log(S_T/S_0)$. $RV(X,\mathcal{P})$ is the discrete realized variance ($\sum_{i=1}^n(\log(S_{t_i}/S_{t_{i-1}}))^2$) and $[X,X]_T$ is the continuous  quadratic variation ($\int_0^T m^2(V_s) ds$ in our setting).

 When $f(x)=x/T$ and the correlation can be positive, the conjecture is violated, see for example Figure \ref{F1} to \ref{F3} where $K^M_d(n)$ can be below $K_c^M$. When $\rho=0$, the process has conditionally independent increments and satisfies other assumptions in  Keller-Ressel and Griessler \citeyear{KG12}. Proposition \ref{p0}  ensures that $K^M_d(n)\ge K_c^M$, which is consistent with their results.

\section{Numerics \label{S4}}

This section illustrates with numerical examples in the Heston, the Hull-White and the Sch\"obel-Zhu models.

\subsection{Heston and Hull-White models\label{MM}}
Given parameters for the Heston model,  we then choose the parameters in the Hull-White model so that the continuous strikes match. Precisely, we obtain $\mu$ by solving numerically $K_c^H=K_c^{HW}$,
and find $\sigma$ such that the variances of $V_T$ in the respective Heston and the Hull-White models match.
%
%
From \eqref{EVt} and \eqref{VtVs}, the  variance for $V_T$ for the Heston model is given by
\begin{align*}
Var^H(V_T)=\frac{\gamma^2}{2\kappa} (\theta+2 e^{-\kappa T}(V_0 -\theta) +e^{-2\kappa T} (\theta-2V_0)).\notag
\end{align*}
The  variance for $V_T$ for the Hull-White model can be computed using \eqref{EVtHW}
\begin{align*}
Var^{HW}(V_T)=V_0^2 e^{2\mu T}(e^{\sigma^2 T}-1).\notag
\end{align*}

The parameters for the Heston model are taken from reasonable parameter sets in the literature. Precisely the first set of parameters is similar to the one used by Broadie and Jain \citeyear{BJ08}. 
 The second set corresponds to Table 2 in Broadie and Kaya \citeyear{BK06}. The values for the parameters of the Hull-White model are obtained consistently using the procedure described above\footnote{For the two sets of parameters above, we compute the critical interest rate $r^*$ as defined in Remark \ref{INTR}. Set 1: $r^*=0.88\%$; Set 2: $r^*=0.605\%$, and we can see that the interest rates are both larger than $r^*$. }.

\begin{table}[!h]
\begin{center}
\begin{tabular}{|c|cccc|ccc|cc|}
\multicolumn{5}{c}{ }&\multicolumn{3}{c}{ }&\multicolumn{2}{c}{(matched)}\\
\multicolumn{5}{c}{ }&\multicolumn{3}{|c|}{Heston }&\multicolumn{2}{|c|}{Hull-White}\\
\hline& $T$& $r$& $V_0$& $\rho$ &  $\gamma$ & $\theta$&$\kappa$&$\mu$& $\sigma$\\\hline
Set 1&1&3.19\%&0.010201&-0.7&0.31&0.019&6.21&1.003&0.42\\
\hline
Set 2&5&5\%&0.09&-0.3&1&0.09&2&$2.9\times10^{-9}$&0.52\\
\hline
\end{tabular}
\caption{Parameter sets\label{t1}}
\end{center}
\end{table}

\begin{center}
{\it Insert Figure \ref{F1}}
\end{center}
Figure \ref{F1} displays cases when the fair strike of the discrete variance swap $K^{M}_d(n)$ may be smaller than the fair strike of the continuous variance swap $K_c^M$. The first graph obtained in the Heston model (the model $M$ is denoted by the exponent $H$ for Heston) shows that $K^{H}_d$ is first higher than $K_c^H$, crosses this level and stays below $K^H_c$ until it converges to the value $K^H_c$ as $n\rightarrow \infty$. It means that  options on discrete realized variance may be overvalued  when the continuous quadratic variation is used to approximate the discrete realized variance.  Note that this unusual pattern happens when $\rho=0.7$, which may happen for example in foreign exchange markets.

\begin{center}
{\it Insert Figure \ref{F1bis}}
\end{center}
Figure \ref{F1bis} highlights another type of convergence showing the complexity of the behaviour of the fair strike of the discrete variance swap with respect to that of the continuous variance swap.

\begin{center}
{\it Insert Figure \ref{F2}}
\end{center}

Figure \ref{F2} displays on the same graphs the discrete fair strike $K_d(n)$ and the first two terms of the expansion formula $K_c^H+\frac{a_1^H}{n}$ for the Heston model and $K_c^{HW}+\frac{a_1^{HW}}{n}$ for the Hull-White model (see Propositions \ref{asympt_n_Heston} and \ref{asympt_n_HW} for the exact expressions of $a_1^H$ and $a_1^{HW}$). It shows that the first term of this expansion is already highly informative as it clearly appears to fit very well for small values of $n$ in both models.

\begin{center}
{\it Insert Figure \ref{F3}}
\end{center}

Figure \ref{F3} further illustrates that the discrete fair strike (for a daily monitoring) can be lower than the continuous fair strike as $K_d^M-K_c^M$ may be negative for high values of the correlation coefficient both in the Heston and the Hull-White models. In Remark \ref{RH} and \ref{RHW}, it is noted  that the first term in the asymptotic expansion with respect to $n$ is linear in $\rho$. From Figure \ref{F2} it is clear that the first term has an important explanatory power. This justifies  the linear behavior observed in Figure \ref{F3} of the difference between discrete and continuous fair strikes with respect to $\rho$. Computations of $\rho_0^H$ and $\rho_0^{HW}$ for each of the risk-free rate levels $r=0\%$, $r=3.2\%$ and $r=6\%$ confirm that it is always positive when $r=0\%$ (which is consistent with Proposition \ref{p0}) and that it can be higher than 1, which ensures that for $n$ sufficiently high, the discrete fair strike is always higher than the continuous fair strike.


\begin{center}
{\it Insert Figure \ref{F4}}
\end{center}

Figure \ref{F4} shows that as the time to maturity $T$ goes to $0$, the discrete fair strike is converging to the continuous fair strike at approximately a quadratic rate. This is consistent with Proposition \ref{asympt_T_Heston} and Proposition \ref{asympt_T_HW}.

\begin{center}
{\it Insert Figure \ref{F5}}
\end{center}

Figure \ref{F5} shows that the discrepancy between the discrete fair strike and the continuous fair strike is exacerbated by the volatility of the underlying variance process. We observe that the gap between the discrete fair strike and the continuous fair strike, with respect to $\gamma$, is wider in the Heston model than in the Hull-White model. This illustrates, from a numerical viewpoint, that the discrete fair strike in the Heston model is more sensitive to the volatility of variance parameter than that of the Hull-White model. In particular, the continuous fair strike $K_c^H$ is independent of $\gamma$. For each $\gamma$ we compute the corresponding $\sigma$ for the Hull-White model such that the variances match as described in Section \ref{MM}. We then observe similar patterns in the Heston and the Hull-White models. From the left panel of Figure \ref{F5}, we can see that the shape of the discrete fair strike in the Heston model with respect to $\gamma$ evolves similar to a parabola, and this is consistent with Proposition \ref{asympt_gamma_Heston}. The right panel of Figure \ref{F5} is consistent with Proposition \ref{asympt_sigma_HW}.

\subsection{Sch\"obel-Zhu model\label{SZN}}

For the Sch\"obel-Zhu model, we reproduce a similar numerical analysis and take parameters consistent with the Heston model. Note that the $V$ process in the Sch\"obel-Zhu model corresponds to the volatility process instead of the variance process\footnote{The notation $V_t$ in the Sch\"obel-Zhu model corresponds to the square root of what is denoted by $V_t$ in the Heston model.}. Then we choose $\theta=\sqrt{0.019}$ and $V_0=\sqrt{0.010201}$. Other parameters are taken from set 1 of  Table \ref{t1}.
\begin{center}
{\it Insert Figure \ref{F6}}
\end{center}

Both the left and right panels of Figure \ref{F6} show that $K^{SZ}_d$ can be below $K_c^{SZ}$ until it converges to the value $K^{SZ}_c$ as $n\rightarrow \infty$. This unusual pattern happens when the correlation is positive similarly in the Heston and the Hull-White models. 

\begin{center}
{\it Insert Figure \ref{F7}}
\end{center}

Figure \ref{F7} illustrates that the discrete fair strike (for a daily monitoring) can be lower than the continuous fair strike as $K_d^{SZ}-K_c^{SZ}$ may be negative for high values of the correlation coefficient.  From Figure \ref{F7} it is clear that the first term also has  an important explanatory power. This justifies the linear behavior observed in Figure \ref{F7} of the difference between discrete and continuous fair strikes with respect to $\rho$. Computations of $\rho_0^{SZ}$ (defined as the zero of $a_1^{SZ}$ computed in Proposition \ref{asympt_n_SZ}) for each of the risk-free rate levels $r=0\%$, $r=3.2\%$ and $r=6\%$ confirm that it is always positive when $r=0\%$ (which is consistent with Proposition \ref{p0}).

\section{Conclusions}
This paper provides explicit expressions of the fair strikes of discretely sampled variance swaps in the Heston, the Hull-White and the Sch{\"o}bel-Zhu models. For the Heston model, the explicit closed-form formula simplifies the expressions  obtained by  Broadie and Jain \citeyear{BJ08} in equations (A-29) and (A-30) on page $793$, where several sums from 0 to $n$ are involved. Our formulae are more explicit (as there are no sums involved in the discrete fair strikes), and easier to use. The explicit closed-form formulas for the Hull-White and  the Sch{\"o}bel-Zhu models are new. Asymptotics of the fair strikes with respect to key parameters such as $n\rightarrow \infty$, $T\rightarrow 0$,  $\gamma\rightarrow 0$ are new and consistent with theoretical results obtained in Keller-Ressel and Muhle-Karbe \citeyear{KM12}.

There are several potential research directions. For example this work can be extended to mixed exponential jump diffusions models (proposed by Cai and Kou \citeyear{CK11} and to the calculations of fair strikes for gamma swaps (see for example Lee \citeyear{L10} for a definition of the payoff). For the $3/2$ stochastic volatility model treated in Itkin and Carr \citeyear{IC}, the difficulty lies in obtaining closed-form expressions of the covariance terms  $E\left[h(V_t) h(V_s)\right]$ for some functional $h$, and the determination of a closed-form formula for the fair strike in this model is left as an open problem. Pricing discrete volatility derivatives with non-linear payoffs (e.g. puts and calls on realized variance) in time-homogeneous stochastic volatility models  is also left as a future research direction.

\newpage
\singlespacing
\appendix

\section{Proof of Proposition \ref{p1} \label{pp1}}

Using It{$\bar{\hbox{o}}$}'s lemma and  Cholesky decomposition, \eqref{eq1} becomes
\begin{eqnarray*}
d\( \ln \(S_t\) \) &=&  \(r-\frac{1}{2} m^2(V_t) \)   dt + \rho m(V_t) dW_t^{(2)}+\sqrt{1- \rho^2} m(V_t) dW_t^{(3)}\notag\\
dV_t &=& \mu (V_t) dt +\sigma(V_t)dW_t^{(2)}.\notag
\end{eqnarray*}
where $W_t^{(2)}$ and $W_t^{(3)}$ are two standard independent Brownian motions.

Proposition \ref{p1} is then a direct application of the following lemma (see Lemma $3.1$  of Bernard and Cui \citeyear{BC11} for its proof).
\begin{lem}\label{l1}
Under the model given in \eqref{eq1}, we have
\begin{multline}
S_T=S_0 \exp\left\{ rT -\frac{1}{2}\int_0^T m^2(V_t) dt+\rho(f(V_T)-f(V_0))\right.\\
\left.-\rho \int_0^T h(V_t) dt + \sqrt{1-\rho^2} \int_0^T
m(V_t) dW_t^{(3)} \right\},\label{eqgsv}
\end{multline}
where
$f(v)=\int_0^v \frac{m(z)}{\sigma(z)}dz$ and
$h(v)=\mu(v)f^{\prime}(v)+\frac{1}{2}\sigma^2 (v)f^{\prime\prime}(v).$
\end{lem}

Now from equation \eqref{eqgsv} in Lemma \ref{l1}, we compute the following key elements in the fair strike of the discrete variance swap. Assume that the time interval is $[t, t+\Delta]$, then
\begin{align*}\ln\( \frac{S_{t+\Delta}}{S_t}\)&= r\Delta -\frac{1}{2}\int_t^{t+\Delta} m^2(V_s)ds +\rho \( f(V_{t+\Delta})-f(V_t)-\int_t^{t+\Delta} h(V_s)ds   \)\notag\\
&\ \ \ +\sqrt{1-\rho^2} \int_t^{t+\Delta} m(V_s) dW_s^{(3)}.\notag
\end{align*}
Then we can compute
\begin{align}
&\E\left[\(\ln \frac{S_{t+\Delta}}{S_t}\)^2\right] = r^2 \Delta^2 +\frac{1}{4}\E\left[  \(\int_t^{t+\Delta} m^2(V_s)ds\)^2\right]-r\Delta \E\left[ \int_t^{t+\Delta} m^2(V_s)ds \right]+\E\left[A^2\right]\notag\\
&\quad\quad+\E\left[\(2r \Delta -\int_t^{t+\Delta} m^2(V_s)ds \)A\right] +(1-\rho^2)\E\left[ \int_t^{t+\Delta} m^2(V_s)ds \right],\label{g1}
\end{align}
where $A=\rho \( f(V_{t+\Delta})-f(V_t)-\int_t^{t+\Delta} h(V_s)ds   \),$ and
\begin{align*}
A^2 &=\rho^2 \( (f(V_{t+\Delta})-f(V_t))^2 +   \(\int_t^{t+\Delta} h(V_s)ds\)^2 -2(f(V_{t+\Delta})-f(V_t))\int_t^{t+\Delta} h(V_s)ds \).\notag
\end{align*}
Using the above expressions for $A$ and $A^2$ in \eqref{g1}, we obtain
\begin{align}
&\E\left[\(\ln \frac{S_{t+\Delta}}{S_t}\)^2\right] \notag\\
&=r^2 \Delta^2 +\frac{1}{4}\E\left[  \(\int_t^{t+\Delta} m^2(V_s)ds\)^2\right]+(1-\rho^2 -r\Delta) \E\left[ \int_t^{t+\Delta} m^2(V_s)ds \right]\notag\\
&\ \ \ +\rho^2 \E[((f(V_{t+\Delta})-f(V_t))^2]+\rho^2 \E\left[ \(\int_t^{t+\Delta} h(V_s)ds\)^2 \right]+2r\rho \Delta \E[(f(V_{t+\Delta})-f(V_t))] \notag\\
&\ \ \ - \E\left[  (f(V_{t+\Delta})-f(V_t))\int_t^{t+\Delta} (2\rho^2 h(V_s)+\rho m^2(V_s))ds   \right]-2r\rho\Delta \E\left[ \int_t^{t+\Delta} h(V_s)ds \right]\notag\\
&\ \ \ +\rho \E\left[ \( \int_t^{t+\Delta} h(V_s)ds \)\( \int_t^{t+\Delta} m^2(V_s)ds \)  \right].\label{i1}
\end{align}
By It{$\bar{\hbox{o}}$}'s lemma, $f$ defined in Lemma \ref{l1} verifies
$df(V_t)=h(V_t)dt +m(V_t)dW_t^{(2)}.$
 Integrating the above SDE from $t$ to $t+\Delta$, we have
\begin{align}
f(V_{t+\Delta})-f(V_t)&=\int_t^{t+\Delta} h(V_s)ds+ \int_t^{t+\Delta} m(V_s)dW_s^{(2)}.\notag
\end{align}
Thus
\begin{align}
\E\left[f(V_{t+\Delta})-f(V_t)\right]-\E\left[\int_t^{t+\Delta} h(V_s)ds\right]&=\E\left[\int_t^{t+\Delta} m(V_s)dW_s^{(2)}\right]=0.\label{i2}
\end{align}
Rearrange \eqref{i1} and use \eqref{i2} to simplify the terms, and we obtain
\begin{multline}
\E\left[\(\ln \frac{S_{t+\Delta}}{S_t}\)^2\right]=r^2 \Delta^2-r\Delta\E\left[\int_t^{t+\Delta}  m^2(V_s)ds\right]
+\frac{1}{4}\E\left[  \(\int_t^{t+\Delta} m^2(V_s)ds\)^2\right]\\
+(1-\rho^2 ) \E\left[ \int_t^{t+\Delta} m^2(V_s)ds\right]+\rho^2 \E\left[\(f(V_{t+\Delta})-f(V_t)\)^2\right]\\
+\rho^2 \E\left[ \(\int_t^{t+\Delta} h(V_s)ds\)^2 \right]
 +\rho \E\left[  \int_t^{t+\Delta} h(V_s)ds \  \int_t^{t+\Delta} m^2(V_s)ds   \right]\\
 - \rho\E\left[  (f(V_{t+\Delta})-f(V_t))\int_t^{t+\Delta} (2\rho h(V_s)+ m^2(V_s)) ds   \right].\label{i3}
\end{multline}
Now we apply Fubini's theorem and partial integration to further simplify \eqref{i3}. Note that  $m^2(V_s)\geq 0$, Q-a.s., then by Fubini's theorem for non-negative measurable functions, $\E\left[\int_t^{t+\Delta}  m^2(V_s)ds\right]=\int_t^{t+\Delta} \E\left[ m^2(V_s)\right] ds$.
Similarly we have $\E\left[  \(\int_t^{t+\Delta} m^2(V_s)ds\)^2\right]=  \int_t^{t+\Delta}\int_t^{t+\Delta} \E\left[m^2(V_s)m^2(V_u)\right]dsdu$ for any $t\leq s\leq t+\Delta$ and any $t\leq u\leq t+\Delta$,

If $\E\left[\mid h(V_s)h(V_u)\mid \right]<\infty$ for any $t\leq s\leq t+\Delta$ and any $t\leq u\leq t+\Delta$, then we have $\E\left[ \(\int_t^{t+\Delta} h(V_s)ds\)^2 \right]=  \int_t^{t+\Delta}\int_t^{t+\Delta} \E\left[h(V_s)h(V_u)\right]dsdu$.

If $\E\left[\mid h(V_s)m^2(V_u)\mid \right]<\infty$ for any $t\leq s\leq t+\Delta$ and any $t\leq u\leq t+\Delta$, then we have
\begin{align}
\E\left[ \int_t^{t+\Delta} h(V_s)ds\ \int_t^{t+\Delta} m^2(V_s)ds \right]= \int_t^{t+\Delta}\int_t^{t+\Delta} \E\left[h(V_s)m^2(V_u)\right]dsdu.\notag
\end{align}

If $\E\left[\mid(f(V_{t+\Delta})-f(V_t))(2\rho h(V_s)+ m^2(V_s))\mid\right]<\infty$ for all $t\leq s\leq t+\Delta$, then we have
\begin{align}
&\E\left[  (f(V_{t+\Delta})-f(V_t))\int_t^{t+\Delta} (2\rho h(V_s)+ m^2(V_s)) ds   \right]\notag\\
&\quad=\int_t^{t+\Delta} \E\left[(f(V_{t+\Delta})-f(V_t))(2\rho h(V_s)+ m^2(V_s))\right] ds.\notag
\end{align}

Thus we finally have proved \eqref{g4} from Proposition \ref{p1}. This completes the proof. \hfill$\Box$

\section{Proof of Proposition \ref{p2} \label{p2p}}
\proof We apply Proposition \ref{p1} to the Heston stochastic volatility model. We first compute $f(x)=\frac{x}{\gamma}$ and $h(x)=\frac{\kappa\theta -\kappa x}{\gamma}$, then we have
\begin{align}
K^{H}_d&=\frac{1}{T}\sum_{i=0}^{n-1} \E\left[\(\ln \frac{S_{t_{i+1}}}{S_{t_i}}\)^2   \right]\notag\\
&=\frac{1}{T}\(\frac{a^2 T^2}{n} +b^2 \sum_{i=0}^{n-1}\int^{t_{i+1}}_{t_i}\int^{t_{i+1}}_{t_i} \E[ V_s V_u] dsdu+\(\frac{2abT}{n}+1-\rho^2\) \int_{0}^{T}\E[V_s] ds\right.\notag\\
&\left.\ \quad \ \ +\frac{\rho^2}{\gamma^2} \sum_{i=0}^{n-1}\E[(V_{t_{i+1}}-V_{t_{i}})^2]+\frac{2\rho aT}{n\gamma} (\E[V_{T}]-\E[V_{0}])\right.\notag\\
&\left.\ \ \ \quad +\frac{2\rho b}{\gamma} \sum_{i=0}^{n-1}\( \int^{t_{i+1}}_{t_i} \E\left[V_{t_{i+1}} V_s \right] ds- \int^{t_{i+1}}_{t_i} \E\left[ V_{t_{i}} V_s \right] ds  \)\). \label{kd2d}
\end{align}
Furthermore, for all $t\ge0$
\be\label{EVt}
\E[V_t]=\theta+e^{-\kappa t}(V_0-\theta),\label{Vt}
\ee
and for all $0<s\leq t$
\begin{align}
\E[V_tV_s]=&\ {\theta}^{2}+e^{-\kappa t}(V_0-\theta)\(\theta+\frac{\gamma^2}{\kappa}\)+e^{-\kappa s}\theta(V_0-\theta)\notag\\
&+e^{-\kappa(t+s)}\((\theta-V_0)^2+\frac{\gamma^2}{2\kappa}(\theta-2V_0)\)+\frac{\gamma^2}{2\kappa}\theta e^{-\kappa(t-s)}.\label{VtVs}
\end{align}
In particular, this formula holds for $t=s$ and gives $\E[V_t^2]$. These formulas already appear in Broadie and Jain \citeyear{BJ08} (formula (A-15)). To compute $K^{H}_d$, \eqref{Vt} and \eqref{VtVs} are the only expressions needed, and they should then be integrated and summed.

We have computed all terms in \eqref{kd2d} with the help of Maple and also have simplified the final expression given by Maple. It turns out that in the case of the Heston model, all terms can be computed explicitly and the final simplified expression for \eqref{kd2d} does not require any sums or integrals.  We finally obtain an explicit formula for $K^{H}_d$ as a function of the parameters of the model. This completes the proof.
\hfill$\Box$

\section{Proof of Proposition \ref{zl}\label{proofzl}}

\proof Denote the log stock price without drift as $X_t =\ln S_t -rt$, and $X_0 =x_0$. Denote $V_0 =v_0$, $\Delta=T/n$.  We have that $\E\left[\(\frac{S_{t_{i+1}}-S_{t_{i}}}{S_{t_{i}}}\)^2\right]=\E\left[\(\frac{S_{t_{i+1}}}{S_{t_{i}}}\)^2\right]+1-2e^{r\Delta}$. Thus the goal is to calculate the second moment $\E\left[\(\frac{S_{t_{i+1}}}{S_{t_{i}}}\)^2\right]$, and note that it is closely linked to the moment generating function of the log stock price $X$. Recall the following formulation of the moment generating function $M(u,t)=\E[e^{u X_t }]$ from Albrecher et al. \citeyear{A}
\begin{align}
M(u,t)
&=S_0^u \exp\left\{ \frac{\kappa\theta}{\gamma^2} \( (\kappa-\gamma \rho u-d(u))t-2\ln\( \frac{1-g(u) e^{-d(u) t}}{1-g(u)}\)  \)  \right\} \notag\\
&\times \exp\left\{ V_0 \frac{\kappa-\gamma \rho u-d(u)}{\gamma^2} \frac{1-e^{-d(u) t}}{1-g(u) e^{-d(u) t}} \right\},\label{mequ}
\end{align}
where the auxiliary functions are given by
\begin{align}
d(u)&=\sqrt{(\kappa-\gamma \rho u)^2 +\gamma^2 (u-u^2)  },\quad
g(u)=\frac{\kappa-\gamma \rho u-d(u)}{\kappa-\gamma \rho u+d(u)}.\notag
\end{align}
We first separate out the case of $i=0$ and $i=1,...,n-1$. For the first case, we have
\begin{align}
\E\left[\(\frac{S_{t_{1}}}{S_{0}}\)^2\right] &=\frac{1}{S_0^2}\E\left[e^{2\ln S_{t_1}}   \right]=\frac{e^{2r t_1}}{S_0^2} M(2, t_1)=\frac{e^{2r \Delta}}{S_0^2} M(2, \Delta).\label{second}
\end{align}

For the second case, with $i=1,2,...,n-1$, we have
\begin{align}
\E\left[\(\frac{S_{t_{i+1}}}{S_{t_i}}\)^2\right]&=\E\left[e^{2\ln \(\frac{S_{t_{i+1}}}{S_{t_i}}\)}\right]=e^{2r \Delta} \E\left[\E\left[e^{2(X_{t_{i+1}}-X_{t_i})} \mid \mathcal{F}_{t_i} \right] \right]\notag\\
&=\exp\left\{2r \Delta+ \frac{\kappa\theta}{\gamma^2} \( (\kappa-2\gamma \rho -d(2))\Delta-2\ln \frac{1-g(2) e^{-d(2)\Delta}}{1-g(2)}  \)  \right\} \notag\\
&\times \E\left[ \exp\left\{V_{t_i} \frac{\kappa-2\gamma \rho -d(2)}{\gamma^2} \frac{1-e^{-d(2)\Delta}}{1-g(2) e^{-d(2)\Delta}} \right\}\right].\label{third}
\end{align}

 We first define $\alpha=2\kappa\theta/\gamma^2 -1 \geq 0$, and $\eta(t) =\frac{2\kappa}{\gamma^2} (1-e^{-\kappa t})^{-1}$. Then from Theorem $3.1$\footnote{Note that in terms of our notation, the parameters in Hurd and Kuznetsov \citeyear{HK08} and our parameters have the correspondence $a=\kappa\theta, b=\kappa,  c=\gamma$.} in Hurd and Kuznetsov \citeyear{HK08}, we have
\begin{align}
\E[e^{u V_T}]&= \(\frac{\eta(T)}{\eta(T) -u}\)^{\alpha+1} e^{ V_0 \frac{\eta(T)u}{\eta(T) -u}e^{-\kappa T}}. \label{fourth}
\end{align}

Combine equations \eqref{third} and \eqref{fourth},  for $i=1,...,n-1$, we finally have
\begin{equation}
\E\left[\(\frac{S_{t_{i+1}}}{S_{t_i}}\)^2\right]
=e^{2r\Delta +  \frac{\kappa\theta}{\gamma^2} \( (\kappa-2\gamma \rho -d(2))\Delta-2\ln \frac{1-g(2) e^{-d(2)\Delta}}{1-g(2)}  \)} e^{V_0 \frac{\eta(t_i) q(2)}{\eta(t_i) -q(2)}e^{-\kappa t_i}} \(\frac{\eta(t_i)}{\eta(t_i) -q(2)}\)^{\alpha+1},\label{fifthequ}
\end{equation}
where $q(u)=\frac{\kappa-\gamma \rho u-d(u)}{\gamma^2} \frac{1-e^{-d(u)\Delta}}{1-g(u)e^{-d(u)\Delta}}$. Using the definition of $M(u,t)$, we can factor out $M(2,\Delta)$ from \eqref{fifthequ} and finally we have
\begin{equation}
\E\left[\(\frac{S_{t_{i+1}}}{S_{t_i}}\)^2\right]
=\frac{e^{2r\Delta}}{S_0^2} M(2, \Delta) e^{ q(2) V_0 \(\frac{\eta(t_i) e^{-\kappa t_i} }{\eta(t_i) -q(2)}-1\)} \(\frac{\eta(t_i)}{\eta(t_i) -q(2)}\)^{\alpha+1}.\label{sixth}
\end{equation}

When $i=0$, we have $t_i =0$ and since $\eta_{u}\rightarrow \infty$ as $u\rightarrow 0$, we use L'H\^{o}pital's rule
\begin{align}
\frac{\eta_{t_0}}{\eta_{t_0}-q(2)}&=\lim\limits_{u\rightarrow 0} \frac{\eta_{u}}{\eta_{u}-q(2)}=\lim\limits_{u\rightarrow 0} \frac{\eta^{\prime}_{u}}{\eta^{\prime}_{u}}=1.\notag
\end{align}
Thus $a_0$ is a special case of the formula in \eqref{sixth} when $i=0$.
From Theorem $3.1$ in Hurd and Kuznetsov \citeyear{HK08}, equation \eqref{mequ} and consequently the above \eqref{third}, \eqref{fourth} are well-defined if $u<\eta(T)$\footnote{Note that $\eta(t)$ is a decreasing function in $t$, thus $u<\eta(T)$ is sufficient for $u<\eta(t_i)$ for all $i=0,1,..., n$.}. Note that the formula \eqref{sixth} involves the  $u=2$ case.  A sufficient condition for $u=2<\eta(T)$ to hold is $\gamma^2 T <1$ (since $2<\eta(T)$ is equivalent to $1-\frac{\kappa}{\gamma^2} <e^{-\kappa T}$ ).

Then the final formula for the discrete fair strike follows by summing the above terms $a_i, i=0,1,...,n-1$. This completes the proof.
\hfill$\Box$

\section{Proof of Proposition \ref{pHW}\label{proofHW}}

\proof
For the Hull-White model, from Proposition \ref{p1}, we first compute $f(x)=\frac{2}{\sigma}\sqrt{x}$ and $h(x)=\(\frac{\mu}{\sigma}-\frac{\sigma}{4}\)\sqrt{x}$, then we have
\begin{align}
&\E\left[\(\ln \frac{S_{t_{i+1}}}{S_{t_i}}\)^2\right]=(1-\rho^2 -\frac{rT}{n}) \int_{\frac{iT}{n}}^{\frac{(i+1)T}{n}} \E\left[ V_s  \right]ds+r^2 \frac{T^2}{n^2} \notag\\
&-\frac{2\rho}{\sigma}  \int_{\frac{iT}{n}}^{\frac{(i+1)T}{n}} \E\left[ \(\sqrt{V_{\frac{(i+1)T}{n}}}-\sqrt{V_{\frac{iT}{n}}}\) V_s \right] ds +2\rho^2 q^2   \int_{\frac{iT}{n}}^{\frac{(i+1)T}{n}} \int_{\frac{iT}{n}}^{u} \E\left[\sqrt{V_s}\sqrt{V_u}\right]dsdu \notag\\
&+\frac{4\rho^2}{\sigma^2} \E\left[\(\sqrt{V_{\frac{(i+1)T}{n}}}-\sqrt{V_{\frac{iT}{n}}}\)^2\right]
 -\frac{4\rho^2 q}{\sigma}  \int_{\frac{iT}{n}}^{\frac{(i+1)T}{n}} \E\left[ \(\sqrt{V_{\frac{(i+1)T}{n}}}-\sqrt{V_\frac{iT}{n}}\)\sqrt{V_s}\right] ds \notag\\
&+\frac{1}{2}  \int_{\frac{iT}{n}}^{\frac{(i+1)T}{n}} \int_{\frac{iT}{n}}^{u} \E\left[V_s V_u\right] dsdu+\rho q  \int_{\frac{iT}{n}}^{\frac{(i+1)T}{n}} \int_{\frac{iT}{n}}^{u} \E\left[ \sqrt{V_s} V_u\right] dsdu\notag\\
&+\rho q  \int_{\frac{iT}{n}}^{\frac{(i+1)T}{n}} \int_{u}^{\frac{(i+1)T}{n}} \E\left[ \sqrt{V_s} V_u\right] dsdu,\notag
\end{align}
with $q=\frac{\mu}{\sigma}-\frac{\sigma}{4}$.


We now compute the following covariance terms that are useful in the simplification of the fair strike $K^{HW}_d(n)$. In the Hull-White model, the stochastic variance process $V_t$ follows a geometric Brownian motion.  Thus we have
$V_t =V_0 \exp\(\(\mu-\frac{\si^2}{2}\)t +\sigma W_t^{(2)}\)$.
Note that
\be\label{EVtHW}\E\left[V_s^a\right]=V_0^ae^{a\mu s}e^{\frac{a^2-a}{2}\sigma^2s},\notag\ee
which will be useful below for $a=1/2$, $a=1$ and $a=2$.
$$\E\left[V_s\right]=V_0e^{\mu s},\quad \E\left[\sqrt{V_s}\right]=\sqrt{V_0}e^{\frac{\mu}{2} s-\frac{1}{8}\sigma^2s}=\sqrt{V_0}e^{\frac{\sigma}{2}qs}, \quad \E\left[V_s^2\right]=V_0^2e^{2\mu s+\sigma^2s}.$$

The fair strike for the continuous variance swap is straightforward and is equal to
$
\E\left[ \int_0^{T} V_s ds \right]=\frac{V_0}{\mu} (e^{\mu T}-1).
$
Similarly
\begin{align}
\E\left[ \int_0^{T} \sqrt{V_s} ds \right]
&=\int_0^{T}\sqrt{V_0} e^{\frac{\sigma}{2}qs}ds=\sqrt{V_0}\frac{2}{\sigma q}\(e^{\frac{\sigma q T}{2}}-1\),\notag
\end{align}
and for $s<u$, we have the following results
\begin{align}
\E\left[V_sV_u\right]&=V_0^2\exp\(\mu(u+s)+\sigma^2 s\),\notag\\
\E\left[\sqrt{V_s}\sqrt{V_u}\right]&=V_0\exp\(\frac{\mu}{2}(u+s)-\frac{\sigma^2}{8} (u-s)\),\notag\\
\E\left[\sqrt{V_s}\,{V_u}\right]&=V_0^{\frac32}\exp\({\mu}\(\frac{s}{2}+u\)+\frac{3\sigma^2}{8} s\),\notag\\
\E\left[{V_s}\,\sqrt{V_u}\right]&=V_0^{\frac32}\exp\({\mu}\(s+\frac{u}{2}\)-\frac{\sigma^2}{8} u+\frac{\sigma^2}{2} s\).\label{lcc2}
\end{align}
After some tedious calculations with the help of Maple, we can obtain an explicit formula as the one appearing in Proposition \ref{pHW}. This completes the proof. \hfill$\Box$

\section{Proof of Proposition \ref{pSZ}\label{proofSZ}}

\proof For the Sch{\"o}bel-Zhu model, from the key equation in Proposition \ref{p1}, we have
\begin{eqnarray}
\E\left[\(\ln \frac{S_{t+\Delta}}{S_t}\)^2\right]&=&r^2 \Delta^2
+(1-\rho^2 -r\Delta ) \int_t^{t+\Delta} m_1(s)ds- \rho \int_t^{t+\Delta} m_5 (t,s) ds\notag\\
&&+\frac{1}{4}  \int_t^{t+\Delta}\int_t^{t+\Delta} m_2(s,u)dsdu+ \frac{\rho^2}{4\gamma^2} \E\left[\(V_{t+\Delta}^2-V_t^2\)^2\right]\label{szp1}\\
&&+\rho^2   \int_t^{t+\Delta}\int_t^{t+\Delta} m_3(s,u) dsdu+\rho \int_t^{t+\Delta}\int_t^{t+\Delta} m_4(s,u)dsdu,\notag
\end{eqnarray}
where \begin{eqnarray}
m_1(s)&:=&\E\left[ m^2(V_s)\right]=\E\left[ V_s^2\right], t\leq s \leq t+\Delta,\notag\\
m_2(s,u)&:=&\E\left[m^2(V_s)m^2(V_u)\right]=\E\left[V_s^2 V_u^2\right], t\leq s\leq t+\Delta,\quad t\leq u \leq t+\Delta,\notag\\
m_3(s,u)&:=&\E\left[ h(V_s) h(V_u)\right], t\leq s\leq t+\Delta,\quad  t\leq u \leq t+\Delta,\label{tech2app}\\
m_4(s,u)&:=&\E\left[ h(V_s) m^2(V_u)\right], t\leq s\leq t+\Delta, \quad t\leq u \leq t+\Delta,\notag\\
m_5(t,s)&:=&\E\left[(f(V_{t+\Delta})-f(V_t))(2\rho h(V_s)+m^2(V_s))\right], t\leq s\leq t+\Delta,\notag
\end{eqnarray}
and $\E\left[\(V_{t+\Delta}^2-V_t^2\)^2\right]=\E\left[V_{t+\Delta}^4\right]+\E\left[V_{t}^4\right]-2\E\left[V_{t+\Delta}^2 V_{t}^2\right]$.
We compute the following two terms in \eqref{szp1} by expanding the products out. For $s\le u$
\begin{align}
m_3(s,u)=&\E\left[\(\frac{\kappa\theta}{\gamma}V_s-\frac{\kappa}{\gamma}V_s^2+\frac{\gamma}{2}   \)\(\frac{\kappa\theta}{\gamma}V_u-\frac{\kappa}{\gamma}V_u^2+\frac{\gamma}{2}   \)\right]\notag\\
=&\E\left[\frac{\kappa^2\theta^2}{\gamma^2} V_s V_u -\frac{\kappa^2\theta}{\gamma^2}(V_s V_u^2+V_s^2 V_u)+\frac{\kappa\theta}{2}(V_s+V_u)\right.\notag\\
&\quad\quad\left.-\frac{\kappa}{2}(V_s^2 +V_u^2)+\frac{\kappa^2}{\gamma^2}V_s^2 V_u^2 +\frac{\gamma^2}{4}\right],\notag
\end{align}
and for $t\le s\le t+\Delta$
\begin{align}
m_5(t,s)=&\frac{1}{2\gamma}\E\left[(V_{t+\Delta}^2-V_t^2)\(2\rho \(\frac{\kappa\theta}{\gamma}V_s-\frac{\kappa}{\gamma}V_s^2+\frac{\gamma}{2}   \)+ V_s^2\)\right]\notag\\
=&\E\left[\frac{\rho\kappa\theta}{\gamma^2}(V_{t+\Delta}^2 V_s -V_t^2 V_s)+\frac{\gamma-2\rho\kappa}{2\gamma^2}(V_{t+\Delta}^2 V_s^2 -V_t^2 V_s^2)\right.\notag\\
&\left.\quad\quad+\frac{\rho}{2}(V_{t+\Delta}^2-V_{t}^2)\right].\notag
\end{align}

It is clear from the above expressions of $m_i$ for $i=1,2,...,5$ that they are all functions of $\E[V_s]$, $\E[V_s^2]$, $\E[V_s^4]$, $\E[V_sV_u]$, $\E[V_s^2V_u]$, $\E[V_sV_u^2]$ and $\E[V_s^2V_u^2]$. We now compute these seven expressions.

\begin{lem}
For the Ornstein-Uhlenbeck process $V$, introduce the auxiliary deterministic functions $e_s:=(V_0-\theta)e^{-\kappa s}+\theta$, and $v(s):=\frac{\gamma^2}{2\kappa}(1-e^{-2\kappa s})$, then
\begin{align}
\E\left[V_s\right]&=e_s,\label{mom1c}\\
\E\left[V_s^2\right]&=e_s^2 +v(s),\label{mom2c}\\
\E\left[V_s^3\right]&=e_s^3+3e_s v(s),\label{mom3c}\\
\E\left[V_s^4\right]&=e_s^4+6e_s^2 v(s) +3v^2(s).\label{mom4c}
\end{align}

For $t\leq s\leq u\leq t+\Delta$
\begin{align}
 \E\left[V_s V_u \right]&=e^{-\kappa (u-s)}\E\left[ V_s^2  \right]+\theta(1-e^{-\kappa(u-s)})\E\left[V_s\right],\notag\\
\E\left[ V_s^2 V_u^2 \right]&=e^{-2\kappa (u-s)}\E\left[V_s^4\right]+2\theta e^{-\kappa (u-s)} (1-e^{-\kappa(u-s)})\E\left[V_s^3\right]\notag\\
&+\(\theta^2(1-e^{-\kappa(u-s)})^2 +\frac{\gamma^2}{2\kappa}(1-e^{-2\kappa(u-s)})  \)\E\left[V_s^2\right].\label{sz3}
\end{align}

For $t\leq s\leq u\leq t+\Delta$
\begin{align}
\E\left[V_s V_u^2 \right]&=e^{-2\kappa (u-s)}\E\left[V_s^3\right]+2\theta e^{-\kappa (u-s)} (1-e^{-\kappa(u-s)})\E\left[V_s^2\right]\notag\\
&+\(\theta^2(1-e^{-\kappa(u-s)})^2 +\frac{\gamma^2}{2\kappa}(1-e^{-2\kappa(u-s)})  \)\E\left[V_s\right].\label{sz6}
\end{align}
For $t\leq s\leq u\leq t+\Delta$
\begin{align}
\E\left[V_s^2 V_u \right]&=e^{-\kappa(u-s)} \E\left[V_s^3\right]+\theta (1-e^{-\kappa(u-s)})\E\left[V_s^2\right].\label{sz7}
\end{align}
\end{lem}
\proof The stochastic variance process $V_s$ follows
\begin{align}
dV_s &=-\kappa(V_s-\theta)ds +\gamma dW_s^{(2)}.\notag
\end{align}
On page 120 of Jeanblanc, Yor and Chesney \citeyear{JYC09}, one finds that the exact solution of the above SDE is
\begin{align}
V_s &=(V_0-\theta)e^{-\kappa s}+\theta +\gamma \int_0^s e^{-\kappa(s-t)}dW^{(2)}_t.\notag
\end{align}
We can compute
\begin{align}
e_s&:=\E\left[V_s\right]=(V_0-\theta)e^{-\kappa s}+\theta,\label{mom1}\\
v(s)&:=Var\left[ V_s \right]=\frac{\gamma^2}{2\kappa}(1-e^{-2\kappa s}),
\end{align}
and the higher moments can also be computed
\begin{align}
\E\left[V_s^2\right]&=e_s^2 +v(s).\label{mom2}\\
\E\left[V_s^3\right]&=e_s^3+3e_s v(s).\label{mom3}\\
\E\left[V_s^4\right]&=e_s^4+6e_s^2 v(s) +3v^2(s).\label{mom4}
\end{align}
For $s\leq u$, $\E\left[V_u \mid V_s\right]=\E\left[ (V_s-\theta)e^{-\kappa (u-s)}+\theta \right]$, and
\begin{align}
\E\left[V_s V_u \right]&=\E\left[ V_s \E\left[ V_u \mid V_s \right] \right]\notag\\
&=\E\left[ V_s ((V_s-\theta)e^{-\kappa (u-s)}+\theta) \right]\notag\\
&= e^{-\kappa (u-s)}\E\left[ V_s^2  \right]+\theta(1-e^{-\kappa(u-s)})\E\left[V_s\right].\notag
\end{align}

Now we can compute the continuous fair strike as
\begin{align}
K_c &=\frac{1}{T} \E\left[ \int_0^T V_s^2 ds   \right]=\frac{1}{T} \int_0^T  \left[((V_0-\theta)e^{-\kappa s}+\theta)^2+\frac{\gamma^2}{2\kappa}(1-e^{-2\kappa s}) \right]ds\notag\\
&=\( (V_0-\theta)^2 -\frac{\gamma^2}{2\kappa} \)\frac{1-e^{-2\kappa T}}{2\kappa T}+2\theta(V_0-\theta)\frac{1-e^{-\kappa T}}{\kappa T}+\theta^2+\frac{\gamma^2}{2\kappa}.
\end{align}

For $s\leq u$
\begin{align}
\E\left[V_s^2 V_u^2\right]&=\E\left[ V_s^2 \E\left[  V_u^2 \mid V_s      \right]  \right]\notag\\
&=\E\left[ V_s^2 \(((V_s-\theta)e^{-\kappa (u-s)}+\theta)^2+\frac{\gamma^2}{2\kappa}(1-e^{-2\kappa (u-s)}) \) \right]\notag\\
&=e^{-2\kappa (u-s)}\E\left[V_s^4\right]+2\theta e^{-\kappa (u-s)} (1-e^{-\kappa(u-s)})\E\left[V_s^3\right]\notag\\
&+\(\theta^2(1-e^{-\kappa(u-s)})^2 +\frac{\gamma^2}{2\kappa}(1-e^{-2\kappa(u-s)})  \)\E\left[V_s^2\right].\notag\\
\E\left[V_s V_u^2\right]&=\E\left[ V_s \E\left[  V_u^2 \mid V_s      \right]  \right]\notag\\
&=\E\left[ V_s ((V_s-\theta)e^{-\kappa (u-s)}+\theta)^2+\frac{\gamma^2}{2\kappa}(1-e^{-2\kappa (u-s)})  \right]\notag\\
&=e^{-2\kappa (u-s)}\E\left[V_s^3\right]+2\theta e^{-\kappa (u-s)}(1-e^{-\kappa(u-s)})\E\left[V_s^2\right]\notag\\
&+\(\theta^2(1-e^{-\kappa(u-s)})^2 +\frac{\gamma^2}{2\kappa}(1-e^{-2\kappa(u-s)})  \)\E\left[V_s\right].\notag\\
\E\left[V_s^2 V_u\right]&=\E\left[ V_s^2 \E\left[  V_u \mid V_s      \right]  \right]\notag\\
&=\E\left[ V_s^2 ((V_s-\theta)e^{-\kappa (u-s)}+\theta)  \right]\notag\\
&=e^{-\kappa(u-s)} \E\left[V_s^3\right]+\theta (1-e^{-\kappa(u-s)})\E\left[V_s^2\right].\notag
\end{align}

In the above expressions, the moments $\E\left[V_s\right]$, $\E\left[V_s^2\right]$, $\E\left[V_s^3\right]$ and $\E\left[V_s^4\right]$ are already calculated in \eqref{mom1}, \eqref{mom2}, \eqref{mom3}, and \eqref{mom4}. Then we can substitute the corresponding inputs into equation \eqref{szp1}, sum up the terms, and obtain $K_d^{SZ}(n)$. This completes the proof. \hfill$\Box$

\vspace{3cm}

\bibliography{dvsh}

\newpage

\begin{figure}[!htbp]
\hspace{-2cm}\begin{tabular}{cc}
\multicolumn{2}{c}{\includegraphics*[height=16cm,width=18cm]{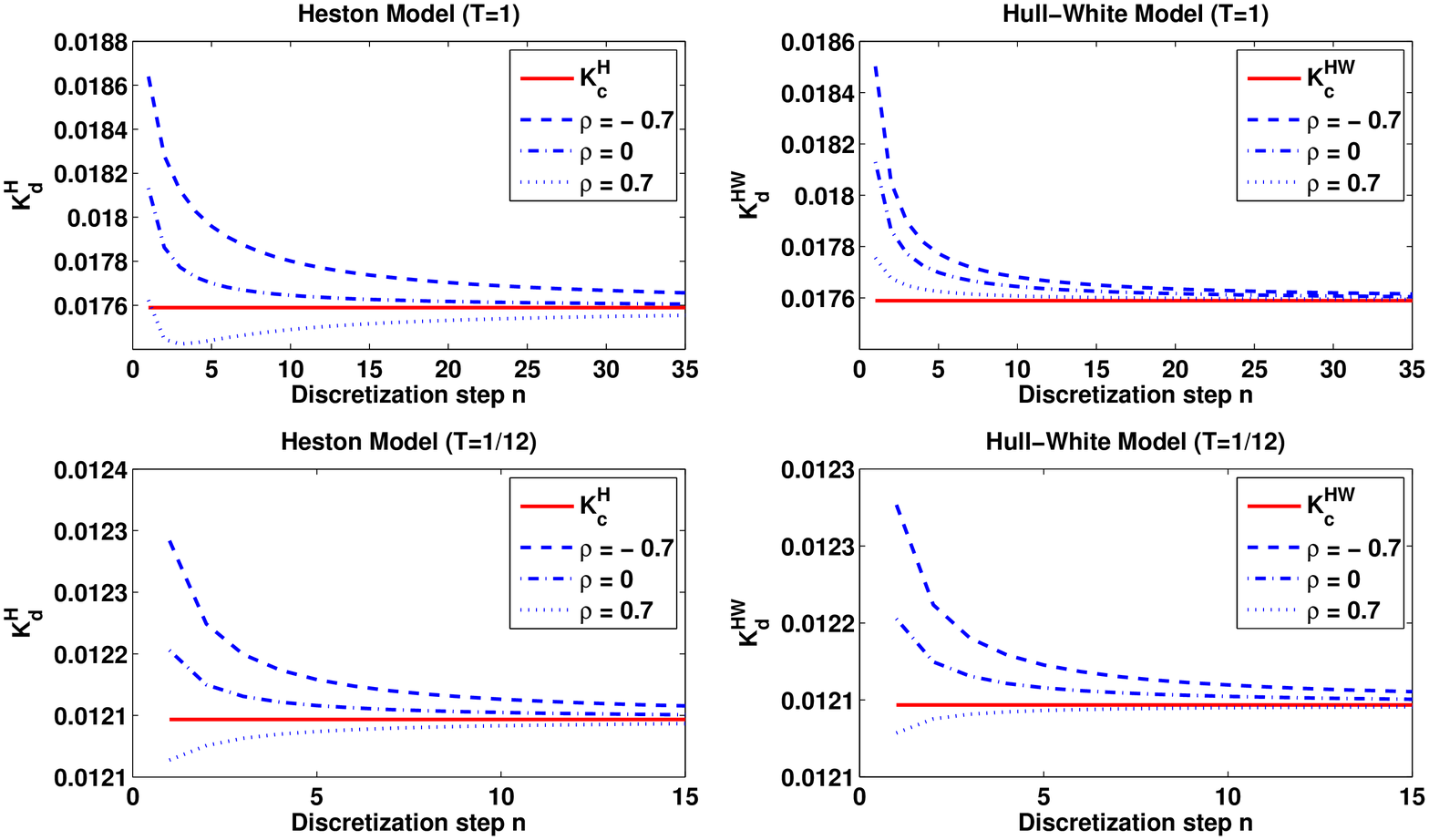}}\\
\end{tabular}
\caption{Sensitivity to the number of sampling periods $n$ and to $\rho$ \label{F1}}\vspace{2mm}
\rule{0pt}{3pt}
\parbox{5in}
{\footnotesize Parameters correspond to Set 1 in Table \ref{t1} except for $\rho$ that can take three possible values $\rho=-0.7$, $\rho=0$ or $\rho=0.7$ and for $T$ that is equal to $T=1$ for the  two upper graphs and  $T=1/12$ for the two lower graphs. When $T=1/12$, the parameters for the Hull-White model are adjusted according to the procedure described in Section \ref{MM}. In the case when $T=1/12$, one has $\mu= 4.03$ and $\sigma=1.78$.
}
\end{figure}

\newpage

\begin{figure}[!htbp]
\hspace{-2cm}\begin{tabular}{cc}
\multicolumn{2}{c}{\includegraphics*[height=8cm,width=18cm]{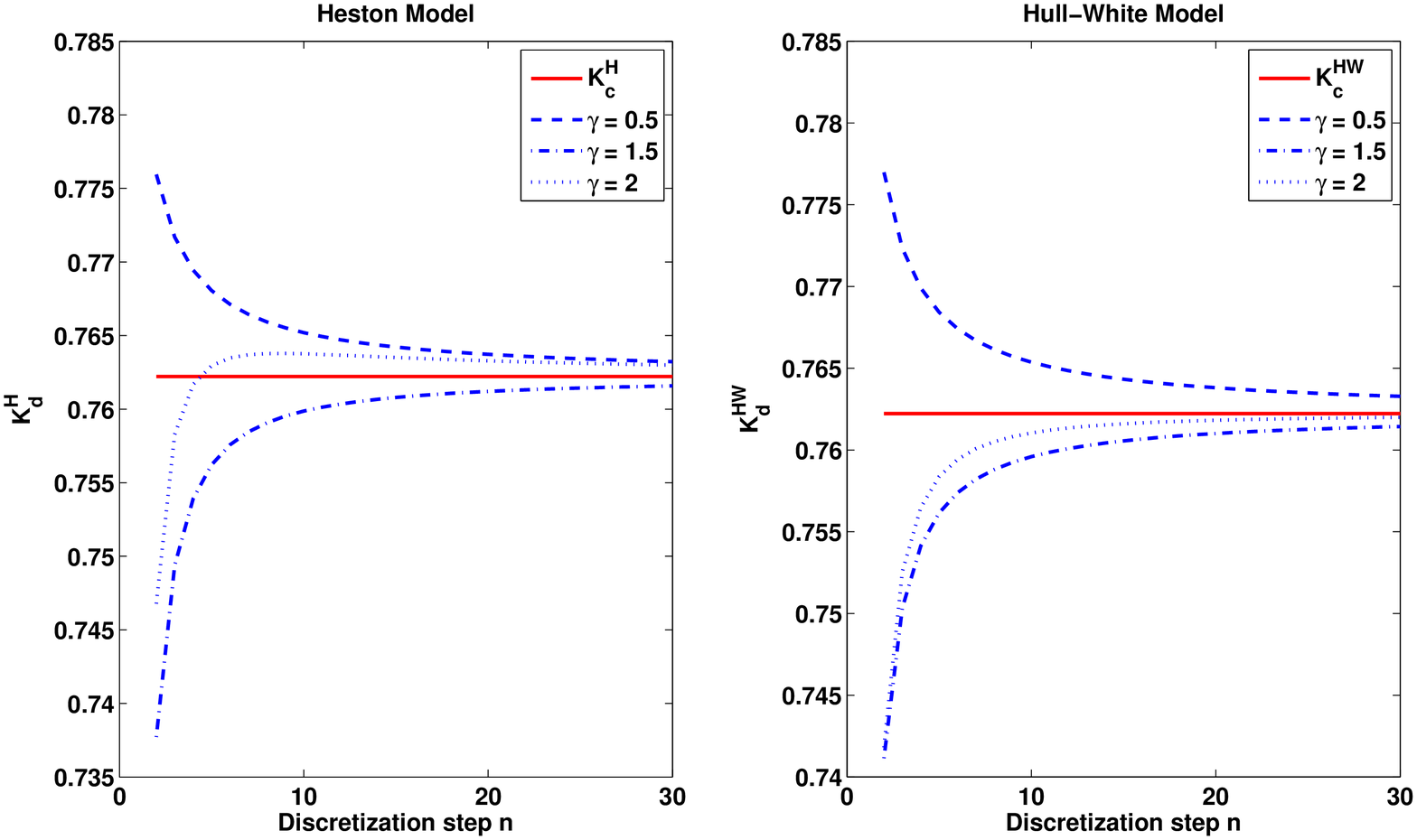}}\\
\end{tabular}
\caption{Sensitivity to the number of sampling periods $n$ and to $\gamma$ \label{F1bis}}\vspace{2mm}
\rule{0pt}{3pt}
\parbox{5in}
{\footnotesize Parameters are set to unusual values to show that any types of behaviors can be expected. $\rho=0.6$, $r=3.19\%$, $\theta=0.019$, $\kappa=.1$, $V_0=0.8$ and $\gamma$ takes three possible values: 0.5, 1.5 and 2.
}
\end{figure}
\newpage

\begin{figure}[!htbp]
\hspace{-3cm}\begin{tabular}{cc}
\multicolumn{2}{c}{\includegraphics*[height=14cm,width=18cm]{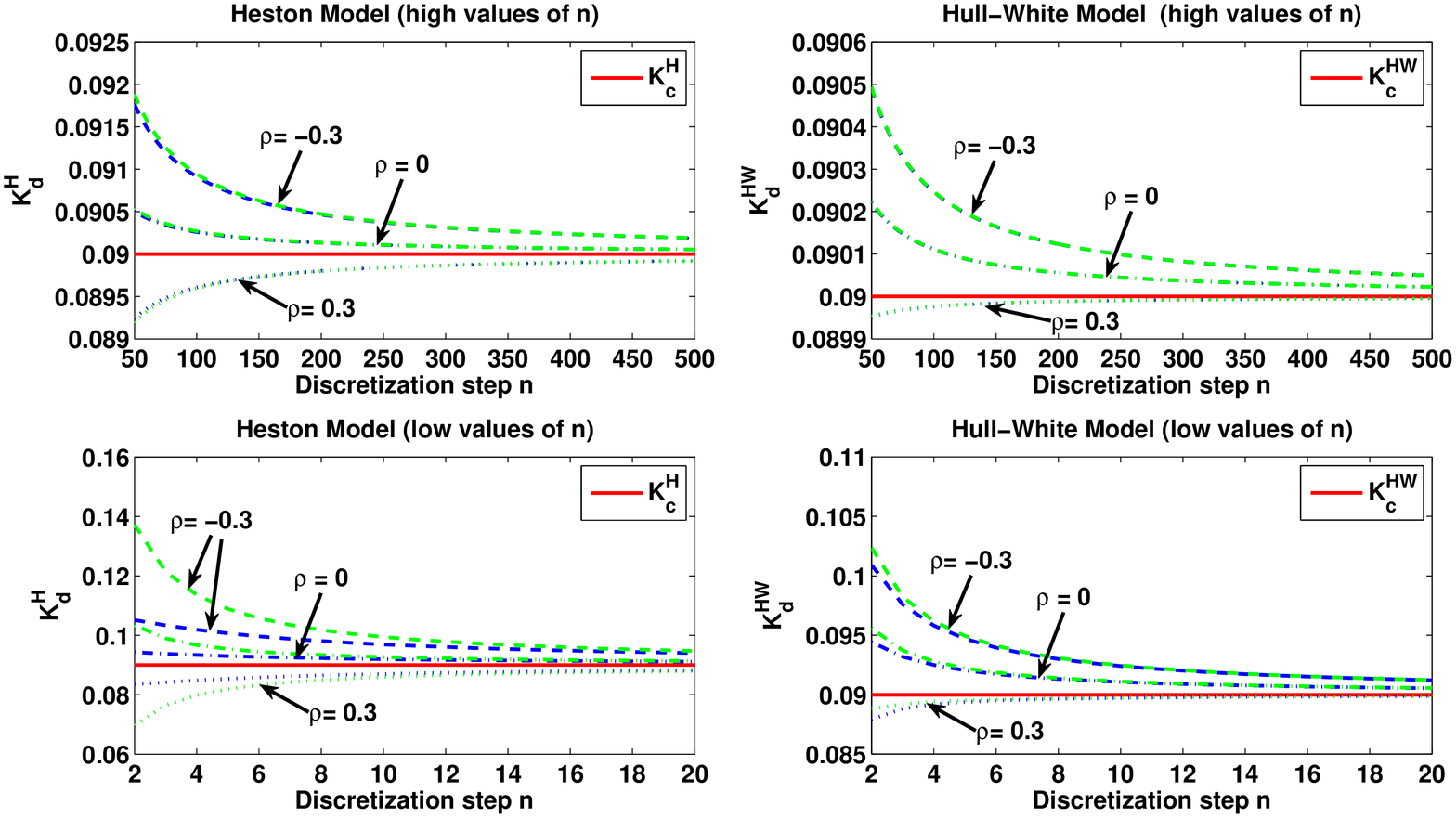}}\\
\end{tabular}
\caption{Asymptotic expansion $K_c^M+a/n$ with respect to the number of sampling periods $n$ and to $\rho$ \label{F2}}\vspace{2mm}
\rule{0pt}{3pt}
\parbox{5in}
{\footnotesize Parameters correspond to Set 2 in Table \ref{t1} except for $\rho$ that can take three possible values $\rho=-0.3$, $\rho=0$ or $\rho=0.3$. The upper graphs correspond to large number of discretization steps whereas lower graphs have relatively small values of $n$.
}
\end{figure}

\newpage
\begin{figure}[!htbp]
\hspace{-3cm}\begin{tabular}{cc}
\multicolumn{2}{c}{\includegraphics*[height=8cm,width=18cm]{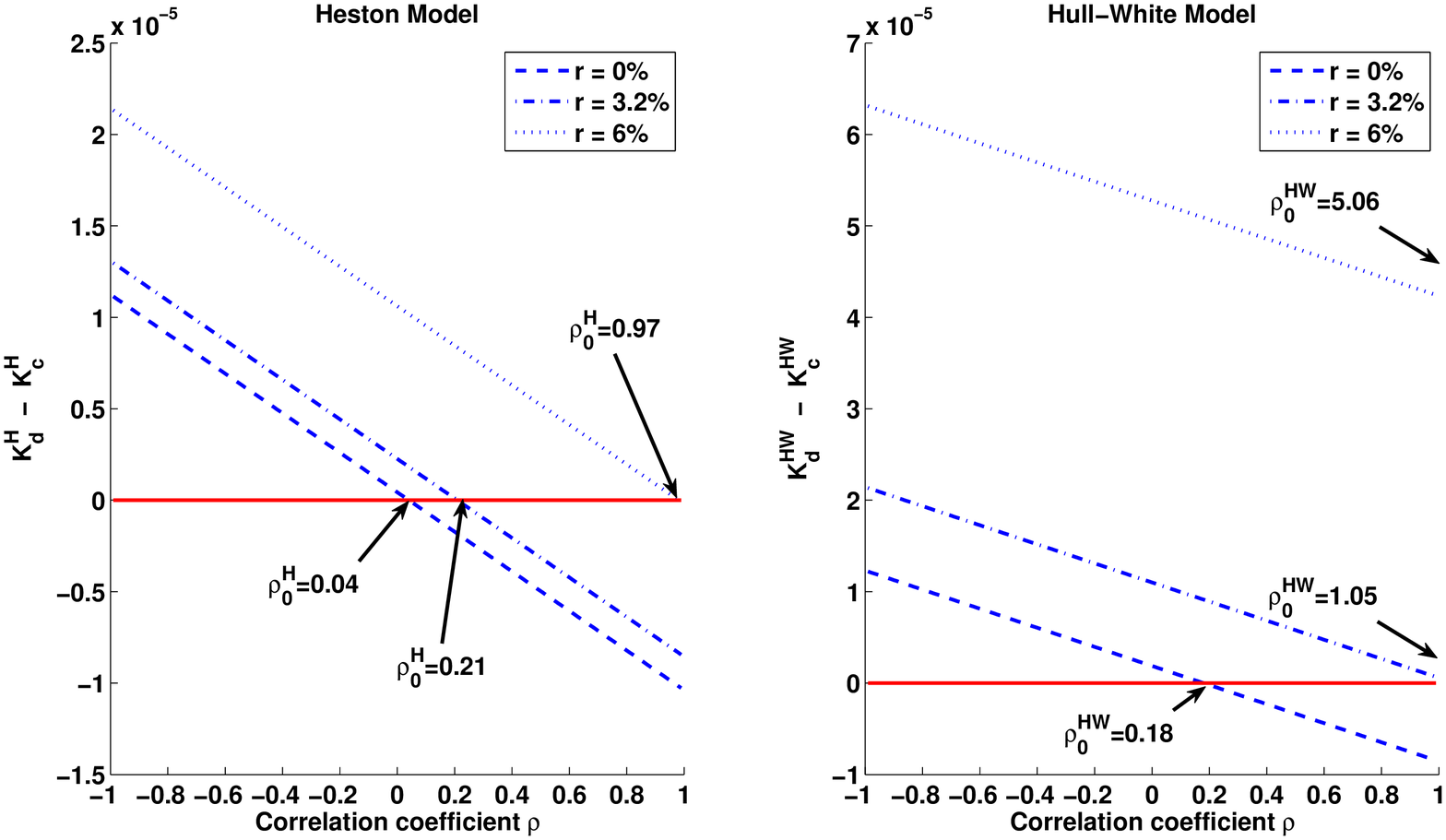}}\\
\end{tabular}
\caption{Asymptotic expansion with respect to the correlation coefficient $\rho$ and the risk-free rate $r$ \label{F3}}\vspace{2mm}
\rule{0pt}{3pt}
\parbox{5in}
{\footnotesize Parameters correspond to Set 1 in Table \ref{t1} except for $r$ that can take three possible values $r=0\%$, $r=3.2\%$ or $r=6\%$. Here $n=250$, which corresponds to a daily monitoring as $T=1$.
}
\end{figure}

\begin{figure}[!htbp]
\hspace{-3cm}\begin{tabular}{cc}
\multicolumn{2}{c}{\includegraphics*[height=8cm,width=18cm]{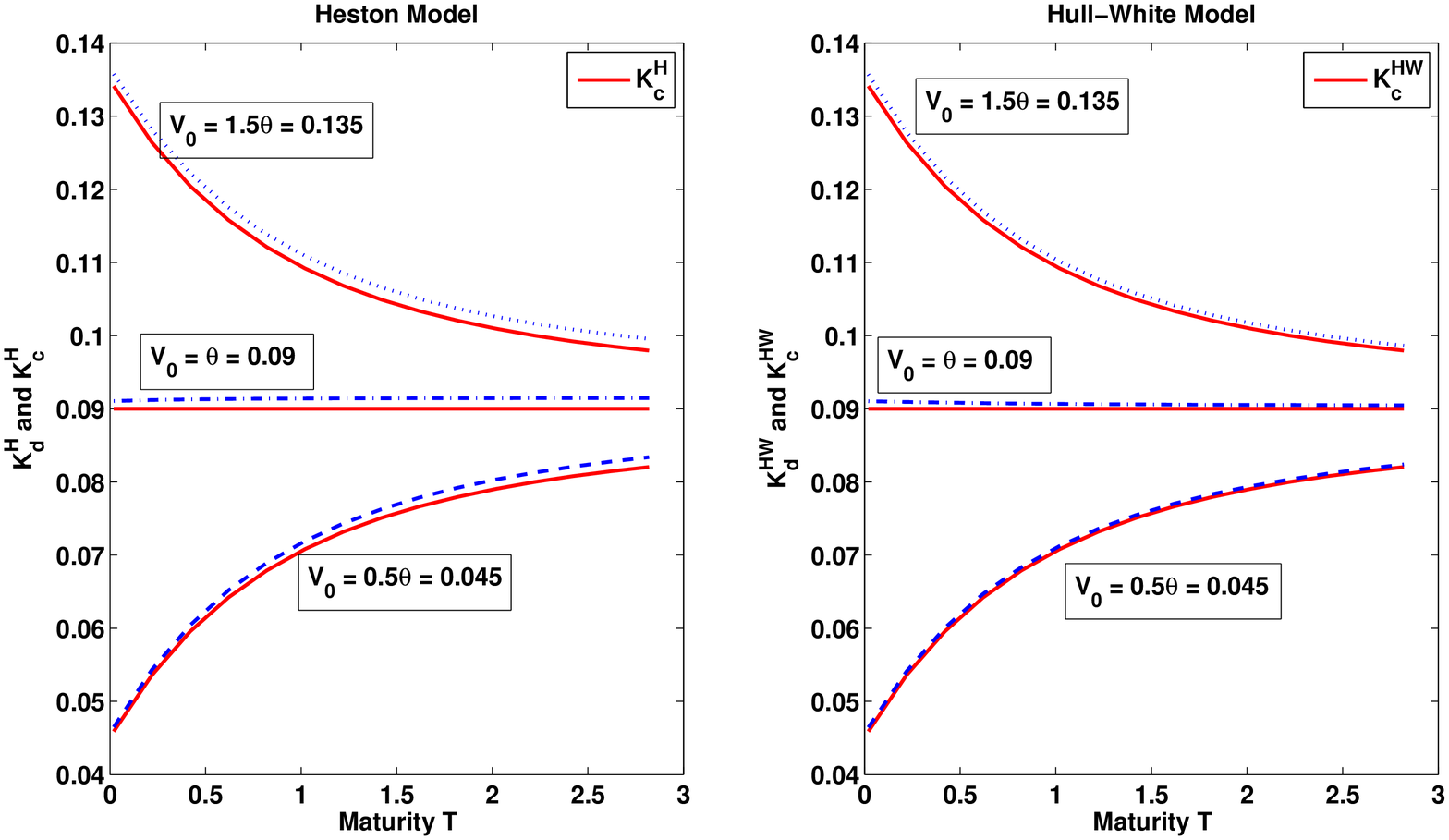}}\\
\end{tabular}
\caption{Discrete and continuous fair strikes with respect to the maturity $T$ and to $V_0$ \label{F4}}\vspace{2mm}
\rule{0pt}{3pt}
\parbox{5in}
{\footnotesize Parameters correspond to Set 2 in Table \ref{t1} except for $T$ and $V_0$. Also we choose a monthly monitoring to compute the discrete fair strike. When $\theta=V_0$, $K_c^H$ is independent of the maturity $T$.
}
\end{figure}

\newpage

\begin{figure}[!htbp]
\hspace{-3cm}\begin{tabular}{cc}
\multicolumn{2}{c}{\includegraphics*[height=8cm,width=18cm]{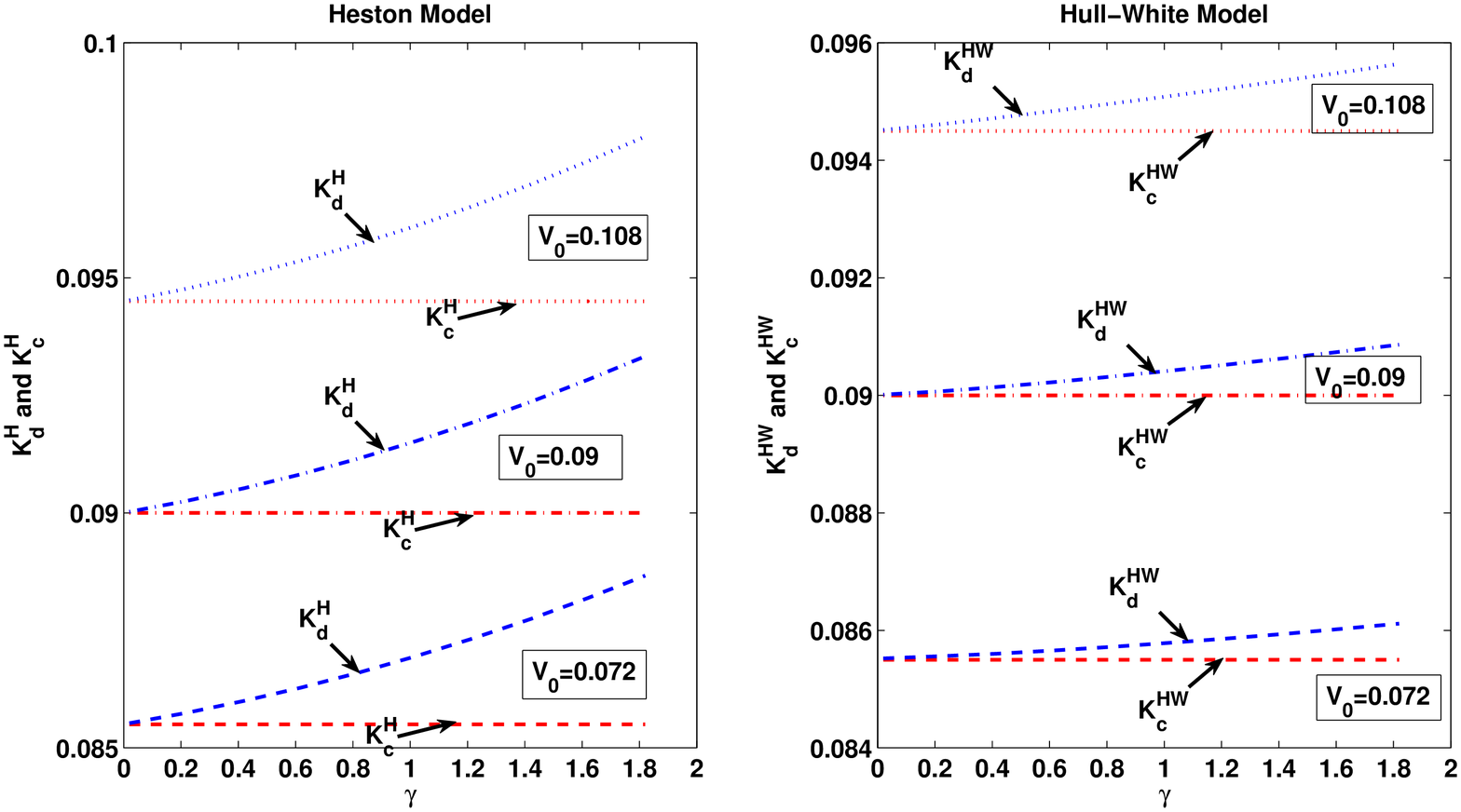}}\\
\end{tabular}
\caption{Discrete and continuous fair strikes with respect to the parameter $\gamma$ and to $V_0$ \label{F5}}\vspace{2mm}
\rule{0pt}{3pt}
\parbox{5in}
{\footnotesize Parameters correspond to Set 2 in Table \ref{t1} except for $\gamma$ and $V_0$ that are indicated on the graphs. A monthly monitoring is used to compute the discrete fair strike. The continuous fair strike $K_c^H$ is independent of $\gamma$, so that it is easy to identify the different curves on the graph.
}
\end{figure}

\newpage

\begin{figure}[!htbp]
\hspace{-3cm}\begin{tabular}{c}
{\includegraphics*[height=8cm,width=18cm]{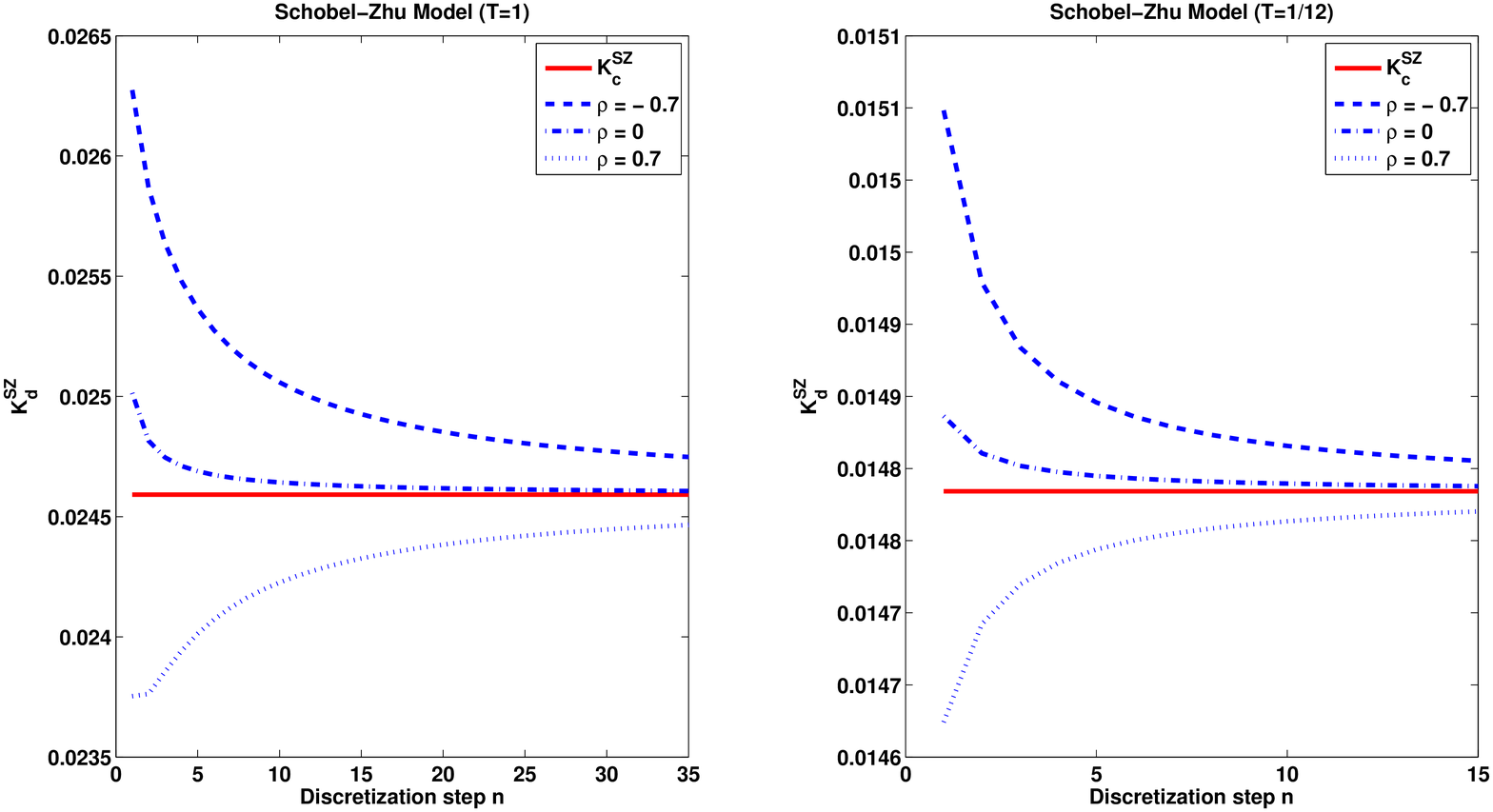}}\\
\end{tabular}
\caption{Sensitivity to the number of sampling periods $n$ and to $\rho$ \label{F6}}\vspace{2mm}
\rule{0pt}{3pt}
\parbox{5in}
{\footnotesize Parameters are similar to Set 1 in Table \ref{t1} for the Heston model except for $\rho$ that can take three possible values $\rho=-0.7$, $\rho=0$ or $\rho=0.7$ and for $T$ that is equal to $T=1$ for the left panel and  $T=1/12$ for the right panel. Precisely, we use the following parameters for the Sch\"obel-Zhu model.
$\kappa=6.21$, $\theta=\sqrt{0.019}$, $\gamma=0.31$, $r=0.0319$, $V_0=\sqrt{0.010201}$.
}
\end{figure}

\newpage

\begin{figure}[!htbp]
\hspace{-1cm}\begin{tabular}{c}
\includegraphics*[height=9cm,width=14cm]{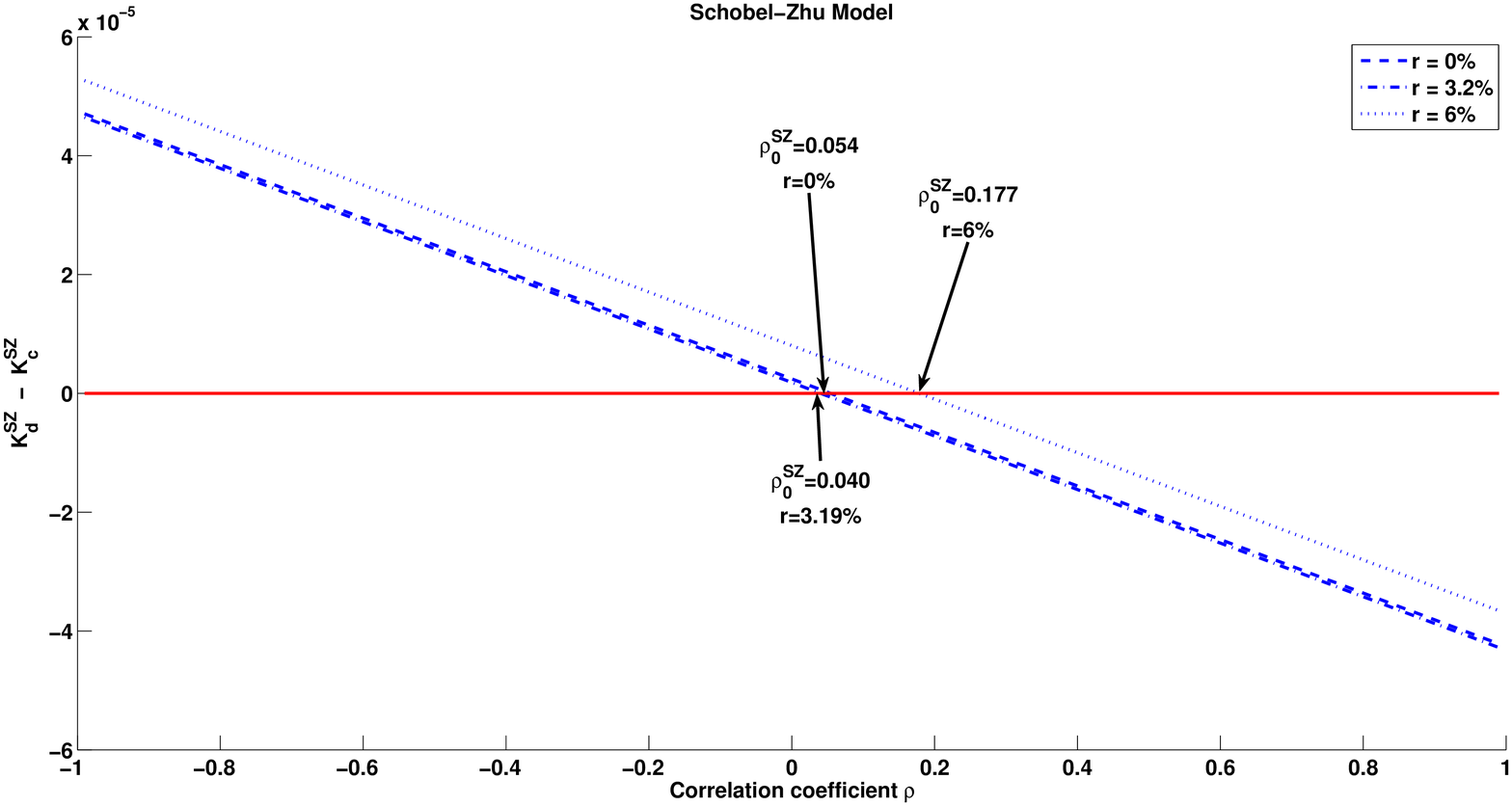}
\end{tabular}
\caption{
Asymptotic expansion with respect to the correlation coefficient $\rho$ and the risk-free rate $r$. \label{F7}}\vspace{2mm}
\rule{0pt}{3pt}\parbox{5in}{\footnotesize  
 Parameters are similar to Set 1 in Table \ref{t1} for the Heston model except for $r$ that can take three possible values $r=0\%$, $r=3.2\%$ or $r=6\%$. Precisely, we use the following parameters for the Sch\"obel-Zhu model: $\kappa=6.21$, $\theta=\sqrt{0.019}$, $\gamma=0.31$, $\rho=-0.7$, $T=1$, $V_0=\sqrt{0.010201}$. Here $n=250$, which corresponds to a daily monitoring as $T=1$.
}
\end{figure}

\end{document}